\documentclass{amsart}
\usepackage{epsfig}
\usepackage{amsmath,amssymb}
\usepackage{slashed}
\usepackage{graphicx}
\usepackage{cite}
\usepackage{tikz-cd} 
\usepackage[colorlinks]{hyperref}

\hypersetup{%
  linkcolor=blue,
  citecolor=red,
  urlcolor=green
}

\numberwithin{equation}{section}

\usepackage{todonotes}

\newtheorem{theorem}{Theorem}[section]

\newtheorem{proposition}[theorem]{Proposition}

\allowdisplaybreaks \numberwithin{equation}{section}

\newcommand{\PP}{{\mathbb P}}
\def\Im{\mathop{\rm Im}\nolimits}
\def\Re{\mathop{\rm Re}\nolimits}

\def\diag{\mathop{\rm Diag}\nolimits}
\def\dim{\mathop{\rm dim}\nolimits}
\def\adj{\mathop{\rm Adj}\nolimits}
\def\Adj{\mathop{\rm  Adj}\nolimits}
\def\im{\mathop{\imath}}

\begin{document}

\title[Construction of Monopoles]{The Construction of Monopoles}
\author{H.W. Braden}
\address{School of Mathematics, Edinburgh University, Edinburgh.}
\email{hwb@ed.ac.uk}
\author{V.Z. Enolski}
\address{National University of \lq\lq Kyiv-Mohyla Academy'', 
Kyiv, Ukraine, 04655.}
\email{venolski@googlemail.com}

\begin{abstract}
We show that the Higgs and gauge fields for a BPS monopole may be constructed directly from the spectral curve without having to solve the gauge constraint needed to obtain the Nahm data. The result is the analogue of the instanton result: given ADHM data one can reconstruct the gauge fields algebraically together with differentiation. Here, given the spectral curve, one can similarly reconstruct the Higgs and gauge fields. This answers a problem that has remained open since the discovery of monopoles.
\end{abstract}

\maketitle

\section{Introduction}
Despite the study of BPS monopoles being a mature subject, now over 35 years old, and having uncovered many remarkable results, a number of the original questions that sparked its development remain unanswered. They are hard. One can ask, for example, what the Higgs and gauge fields or their gauge invariant energy density are, but beyond the spherically symmetric coincident $n$-monopole solution  and some partial results for $SU(2)$ charge $2$ monopoles there are no explicit formulae. Although a number of general methods have been developed to address this, extending earlier analogous results for the construction of instantons (see below), these constructions typically involve a step (such as solving an ODE)
that stymie attempts at an analytic solution.  Numerical results based on these constructions and utilising the increase in computing power over the period has meant that we can understand a number of qualitative  aspects of monopole behaviour; at the very least analytic solutions would give some control over these.
In this paper we shall describe how to explicitly construct the Higgs and gauge fields for a monopole
and circumvent those usually intractable steps; a number of new results will appear in the process. 

We will focus here on $SU(2)$ charge-$n$ monopoles.
Our approach will assume from the outset that we know the spectral curve $\mathcal{C}$ for the monopole. This curve
appears in both Nahm's extension \cite{nahm_80} of the ADHM construction of instantons and Ward's 
\cite{Ward1981a} use of the
$\mathcal{A}_k$ ansatz which Atiyah and Ward used in their instanton construction; it is also implicit in the 
B\"acklund transformation construction of Forg{\'a}cs,  Horv{\'a}th  and Palla \cite{forgacs_horvath_palla_82b}. This curve gives a point in the moduli space of $SU(2)$ charge-$n$ monopoles and we note that the relationship between this description of the moduli space and both Donaldson's rational map and Jarvis's rational map descriptions remains still poorly understood. In this paper we will work within the Nahm construction \cite{nahm82c, hitchin_83}. This involves two
potentially difficult steps: first, the solution
of a first order (matrix) differential equation $\Delta\sp\dagger v=0$; and second, the integration of
appropriate bilinears of these solutions to give the Higgs and gauge fields.
Now the operator $\Delta\sp\dagger$ is in turn constructed from
Nahm data, matrices $T_i(s)$ ($i=1,2,3$) that satisfy Nahm's equations $\dot T_1=[T_2,T_3]$ (and cyclic)
and certain boundary conditions that will be described in more detail below. Unfortunately Nahm data is hard to construct. Some years ago Ercolani and Sinha \cite{ercolani_sinha_89} showed how, using integrable systems techniques, one could solve for a gauge transform of the Nahm data directly from $\mathcal{C}$; this theory has been
further extended by the authors \cite{Braden2010d,Braden2005a}. Determining the gauge transformation is equivalent to solving a further ODE and although its solution exists we don't know how to do this explicitly; we shall relate this to Donaldson's treatment  \cite{donaldson84} of the complex and real moment map description of the Nahm equations. Putting this difficulty to one side for the moment we show how integrable systems techniques and a \emph{lesser known ansatz of Nahm} may be used to solve $\Delta\sp\dagger v=0$, again up to the same gauge transformation. Not all of the
solutions obtained are normalisable and in due course we show that constructing an appropriate projector to these is purely algebraic. But even with these solutions we must perform a number of integrations to
obtain expressions for the sought after fields.  An old work of Panagopoulos \cite{panagopo83} is appropriate here: after some
correction and small extensions we show that not only can all integrations be performed but the as yet undetermined gauge transformation combines within the bilinears into a term that is determined by the spectral curve.  Rather than getting overly involved in the mathematics of Riemann surfaces that any example will necessitate (see for example \cite{Braden2011a,Braden2010b}) this paper will give general formulae for the fields directly in terms of the curve and a well-studied algebro-geometric object, the Baker-Akhiezer function. We conclude with a limited example showing how our construction yields a known result
in the charge $2$ setting; in a sequel paper we shall present the general results for the fields of the charge $2$ monopole.

\section{The ADHMN construction}

In this section we recall the salient features of the ADHMN construction of monopoles sufficient to introduce our notation and the major features of the construction.

The equations we wish to solve are 
\begin{equation} D_i\Phi= \frac12 \sum_{j,k=1}^3\epsilon_{ijk}
F_{jk},\quad i=1,2,3,\label{bogomolny}\end{equation} 
for the gauge group $SU(2)$. Here  $\Phi$ is the Higgs field, $
F_{ij}=\partial_i A_j -\partial_j A_i+[A_i,A_j]$ is the
curvature of the (spatial) connection of the gauge field
$A_i(\boldsymbol{x})$ and $D_i$  the covariant derivative
$D_i\Phi=\partial_i\Phi+[A_i,\Phi]$,
$\boldsymbol{x}=(x_1,x_2,x_3)\in \mathbb{R}^3$.
These equations may be viewed as a reduction of the self-dual Yang Mills equations to three dimensions under the assumption that all fields are independent of time. Upon
identifying the $A_4$-component of the gauge field with the Higgs
field $\Phi $ the four-dimensional Yang-Mills Lagrangian yields 
upon reduction the three dimensional Yang-Mills-Higgs Lagrangian
\[L= -\frac12 \mathrm{Tr}\, F_{ij} F^{ij}+\mathrm{Tr}\,D_{i}\Phi\, D^{i}\Phi.\]
 We are interested
in configurations minimizing the energy of the system. These are
given by the {\it Bogomolny equation} (\ref{bogomolny})
 A solution with the boundary conditions
\[\left. \sqrt{-\frac12 \mathrm{Tr}\,\Phi(r)^2}\right|_{r\rightarrow\infty}\sim
1-\frac{n}{2r}+O(r^{-2}),\quad r=\sqrt{x_1^2+x_2^2+x_3^2} ,\] 
is called a {\it monopole} of  charge $n$.

Modifying the Atiyah-Drinfeld-Hitchin-Manin (ADHM) construction of instanton solutions to the (Euclidean) self-dual Yang-Mills equations  Nahm introduced the operator
\begin{align}\label{defdelta}
\Delta&=\im \dfrac{d}{dz}+x_4-\im T_4+\sum_{j=1}\sp{3}   \sigma_j\otimes( T_j+\im x_j 1_n),
\end{align}
where the $T_j(z)$ are $n\times n$ matrices and $\sigma_j$ the Pauli matrices.
Following the instanton construction the operator $\Delta\sp\dagger\Delta$  must commute with quaternions which happens if and only if ${T_i}\sp\dagger=-T_i$, $T_4\sp\dagger =-T_4$ and
\begin{equation}\label{fullnahm}
\dot{{T}_i} =[T_4,T_i]+\frac12\sum_{j,k=1}^3\epsilon_{ijk}[T_j(s),T_k(s)].
\end{equation}
Equations (\ref{fullnahm}) are known as Nahm's equations; one often encounters them in the more familiar gauge with $T_4=0$. When $\Delta\sp\dagger\Delta$ commutes\footnote{Throughout the superscript $\dagger$ means conjugated and transposed. We will at times emphasise the vectorial nature of an object by 
 printing this in bold,  e.g for vector
$\boldsymbol{a}^{\dagger}=\overline{\boldsymbol{a}}^T$. }
 with quaternions it is a positive operator; in particular this means that $\left(\Delta\sp\dagger\Delta\right)(z)$ is an invertible operator
and consequently $\Delta$ has no zero modes. To describe monopoles the matrices  $T_j(z)$ are further required to be regular for  $z\in(-1,1)$ and have simple poles at $z=\pm1$, the residues of which define an
irreducible $n$-dimensional representation of the $su(2)$ algebra. Hitchin's analysis
\cite{hitchin_83}[\S2]  of the equation 
$\Delta\sp\dagger \boldsymbol{v}=0$ tells us that 
has two normalizable solutions and it is in terms of these that  the
Atiyah-Drinfeld-Hitchin-Manin-Nahm (ADHMN) construction gives the gauge and Higgs field solutions.
\begin{theorem} [\bf ADHMN] \label{ADHMN} The charge $n$
monopole solution of the Bogomolny equation (\ref{bogomolny}) is given by
\begin{align}
\Phi_{ab}(\boldsymbol{x})&=\im\int_{-1}^1\mathrm{d}z\,z \boldsymbol{v}_a\sp
\dagger(\boldsymbol{x},z)\boldsymbol{v}_b(\boldsymbol{x},z)
,\quad a,b=1,2,\label{higgs}\\
A_{i\, ab}(\boldsymbol{x})&=\int_{-1}^1\mathrm{d}z\, \boldsymbol{v}_a\sp\dagger(\boldsymbol{x},z)
\frac{\partial}{\partial x_i} \boldsymbol{v}_b(\boldsymbol{x},z),
\quad
i=1,2,3,\quad a,b=1,2.\label{gauges}
\end{align}
Here the two ($a=1,2$) $2n$-column vectors ${{\boldsymbol{v}}}_{a}(\boldsymbol{x},z)
=({v}_1^{(a)}(\boldsymbol{x},z),\ldots,
{v}_{2n}^{(a)}(\boldsymbol{x},z))^T$ form an orthonormal basis on
the interval $z\in[-1,1]$
\begin{equation}\int_{-1}^1\mathrm{d}z\, \boldsymbol{v}_a\sp\dagger(\boldsymbol{x},z)
\boldsymbol{v}_b(\boldsymbol{x},z)
=\delta_{ab}
,\label{norm}
\end{equation}
for the normalizable solutions to the Weyl equation
\begin{align}\Delta\sp\dagger \boldsymbol{v}=0,\label{deltadagger}
\end{align}
where 
\begin{align}
\Delta^{\dagger}&=\im \dfrac{d}{dz}+x_4-\im T_4-\sum_{j=1}\sp{3}   \sigma_j\otimes( T_j+\im x_j 1_n).
\label{weylequ}\end{align} 
The normalizable solutions form a two-dimensional subspace of the full $2n$-dimensional solution space
to the formal adjoint equation (\ref{deltadagger}).
The $n\times n$-matrices $T_j(z)$, $T_4(z)$, called Nahm data,  satisfy
Nahm's equation (\ref{fullnahm}) and the $T_j(z)$ are
required to be regular for  $z\in(-1,1)$ and have simple poles at $z=\pm1$, the residues of which define an
irreducible $n$-dimensional representation of the $su(2)$ algebra;
further
\begin{equation}
T_i(z)=-T_i^{\dagger}(z),\quad
T_4(z)=-T_4^{\dagger}(z),\quad
T_i(z)=T_i^{T}(-z),\quad
T_4(z)=T_4^{T}(-z)
.\label{constraint}\end{equation}
\end{theorem}

Our strategy will be to solve (\ref{deltadagger}) and determine those solutions that are normalizable.
In what follows we shall denote by $V=(\boldsymbol{v}_1,\ldots,\boldsymbol{v}_{2n})$ and  similarly $W=(\boldsymbol{w}_1,\ldots,\boldsymbol{w}_{2n})$ the $2n\times 2n$ fundamental matrices of solutions to $\Delta\sp\dagger \boldsymbol{v}=0$ and $\Delta \boldsymbol{w}=0$ respectively. Then 
$V$ can be chosen to be $(W^\dagger)^{-1}$.
 As we have already remarked these are not all normalizable.
According to the ADMHN theorem the Nahm data $T_j(z)$ expanded in the vicinity of the end point $z=1-\xi$ behaves as
\begin{equation*}
T_j(1-\xi)=-\im\frac{l_j}{\xi}+O(1),\quad j=1,2,3,
\end{equation*}
where (the Hermitian) $l_j$  define the irreducible $n$-dimensional representation of the $su(2)$ Lie algebra,
$ [l_j,l_k]=\imath\, \epsilon_{jkl}\, l_l $.
Then (\ref{deltadagger})  behaves in the vicinity of the pole as 
\begin{equation}
\left[\frac{\mathrm{d}}{\mathrm{d} \xi}-\frac{\sum_{j=1}^3 \sigma_j\otimes l_j }{\xi}\right]\boldsymbol{v}(\boldsymbol{x},1-\xi)=0.\label{approxdelta}
\end{equation}
One can show (see for example \cite{weinbyi06}) that $\sum_{j=1}^3  \sigma_j\otimes l_j$ has only two distinct eigenvalues, $\lambda_a=(n-1)/2$ with multiplicity $n+1$ and $\lambda_b=-(n+1)/2$
with multiplicity $n-1$
If $\boldsymbol{a}_i$ are eigenvectors associated with $\lambda_a$ ($i=1,\ldots,n+1$), and $\boldsymbol{b}_j$ eigenvectors associated with $\lambda_b$ ($j=1,\ldots,n-1$), then  (\ref{approxdelta}) has solutions
$\xi^{\lambda_a}\boldsymbol{a}_i$ and $\xi^{\lambda_b}\boldsymbol{b}_j$.
Therefore normalizable solutions must lie in the subspace with positive $\lambda_a=(n-1)/2$ and so
we require that $\boldsymbol{v}(\boldsymbol{x},1)$ is orthogonal to the subspace with eigenvalue $-(n+1)/2$, i.e.
\[ \lim_{z\to 1\sp-}\boldsymbol{v}(\boldsymbol{x},z)^T\cdot\boldsymbol{b}_j
=0,\quad j=1,\ldots,n-1.  \]
These $n-1$ conditions coming from the behaviour at $z=1$ thus yield a $n+1$ dimensional space of
solutions to $\Delta\sp\dagger \boldsymbol{v}=0$. A similar analysis at $z=-1$ again yields a further $n-1$ constraints
resulting in two normalisable solutions on the interval.
We define a projector $\mu$ onto this subspace; this is the $2n\times 2$ matrix such that
$$V\mu=(\boldsymbol{v}_1,\boldsymbol{v}_2)$$ where $\boldsymbol{v}_a$ ($a=1,2$) are the normalizable solutions (\ref{norm}).
That is
$$\int_{-1}^1dz\,\mu\sp{\dagger}V\sp{\dagger}V\mu=
\mu\sp{\dagger}\left(\int_{-1}^1dz\,V\sp{\dagger}V\right)\mu
=1_2.$$
Although $V=V(\boldsymbol{x},z)$ and $\boldsymbol{v}_a=\boldsymbol{v}_a(\boldsymbol{x},z)$ are both $z$-dependent we note that the matrix $\mu$ does not depend on $z$ for we have from
$$0=\Delta\sp\dagger (\boldsymbol{v}_1,\boldsymbol{v}_2)=\Delta\sp\dagger (V\mu)=(\Delta\sp\dagger V)\mu
+\im V\dfrac{d}{dz}\mu=\im V\dfrac{d}{dz}\mu$$
and the (generic) invertibility of $V(\boldsymbol{x},z)$  that $\dot\mu=0$ and hence that
$\mu=\mu(\boldsymbol{x})$. Thus to reconstruct the gauge and Higgs fields  we must construct the
projector $\mu$ that extracts from $V$ the two normalizable solutions.

At this stage the integrations in (\ref{ADHMN}) look intractable but work of Panagopoulos enables their evaluation. Define the Hermitian matrices
\begin{equation}\label{pandefs}
\mathcal{H}=-\sum_{j=1}^3 x_j \sigma_j\otimes 1_n,\qquad 
\mathcal{F}=\imath \sum_{j=1}^3 \sigma_j\otimes T_j, \qquad
\mathcal{Q}=\frac{1}{r^2} \mathcal{H}\mathcal{F}\mathcal{H}-\mathcal{F}.
\end{equation}
Then
\begin{proposition}[Panagopoulos \cite{panagopo83}]
 \label{panagopolous} 
\begin{align}\label{pannorm}
 \int \mathrm{d}z\, \boldsymbol{{v}}_a\sp\dagger \boldsymbol{{v}}_b
&= \boldsymbol{{v}}_a\sp\dagger \mathcal{Q}^{-1}
\boldsymbol{{v}}_b.\\
\int \mathrm{d}z\, z \boldsymbol{{v}}_a\sp\dagger \boldsymbol{{v}}_b
&=
\boldsymbol{{v}}_a\sp\dagger\mathcal{Q}^{-1} \left( z+\mathcal{H}\,\frac{x_{i}}{r^2}\frac{\partial}{\partial x_{i}}   \right)\boldsymbol{{v}}_b.
\label{panhiggs}\\
\int \boldsymbol{{v}}_a\sp\dagger\frac{\partial}{\partial
x_i}\boldsymbol{{v}}_b \mathrm{d}z
&=\boldsymbol{{v}}_a\sp\dagger\mathcal{Q}^{-1} \left[
\frac{\partial}{\partial x_i}+\mathcal{H}\frac{z}{r^2}\,
x_i+\mathcal{H}\frac{\imath}{r^2} \left(
\boldsymbol{x}\times\boldsymbol{\nabla}\right)_i \right]\boldsymbol{{v}}_b.\label{pangauge}
\end{align}
\end{proposition}
These are proven in Appendix A.
Interestingly, consideration of gauge invariance leads to new results that will be particularly useful
in our later development.
Recall that a gauge transformation  acts on the normalizable solutions 
$(\boldsymbol{v}_1,\boldsymbol{v}_2)=V\mu$  from the right by $h(\boldsymbol{x})\in SU(2)$.
Then using (\ref{panhiggs})
\begin{align}
\Phi &=i \int_{-1}\sp{1} dz z\, \boldsymbol{v}\sp\dagger \boldsymbol{v}
= i \mu\sp\dagger\left[ \int_{-1}\sp{1} dz z V\sp\dagger V\right]\mu
=i \left[\mu\sp\dagger V\sp\dagger( \mathcal{Q}^{-1} \left( z+\mathcal{H}\,\frac{x_{i}}{r^2}\frac{\partial}{\partial x_{i}}\right)  )V\mu\right]_{z=-1}\sp{z=1}
\label{higgseval}
\end{align}
and so a gauge transformation yields
\begin{align}
h\sp{-1}\Phi h=
h\sp{-1}\left( i \int_{-1}\sp{1} dz z \boldsymbol{v}\sp\dagger \boldsymbol{v} \right) h &=
h\sp{-1}\left( i \mu\sp\dagger\left[ \int_{-1}\sp{1} dz z V\sp\dagger V\right]\mu\right) h
\nonumber
\\
&=h\sp{-1}\left( i \left[\mu\sp\dagger V\sp\dagger( \mathcal{Q}^{-1} \left( z+\mathcal{H}\,
\frac{x_{i}}{r^2}\frac{\partial}{\partial x_{i}}\right)  )V\mu\right]_{z=-1}\sp{z=1}\right) h
\nonumber
\\
&=h\sp{-1}\Phi h+h\sp{-1}
\left( i \left[\mu\sp\dagger V\sp\dagger\mathcal{Q}^{-1} \mathcal{H}V\mu\right]_{z=-1}\sp{z=1}\right) 
\frac{x_{i}}{r^2}\frac{\partial}{\partial x_{i}} h
\nonumber
\\
\intertext{thus we must have}
0&=
\left[\mu\sp\dagger V\sp\dagger\mathcal{Q}^{-1} \mathcal{H}V\mu\right]_{z=-1}\sp{z=1}.
\label{gaugeconsistency}
\end{align}
Further, the transformation of the gauge field
\begin{equation}
A_{i}=\int_{-1}^1\mathrm{d}z\, \boldsymbol{v}\sp\dagger
\frac{\partial}{\partial x_i} \boldsymbol{v}
=\mu\sp\dagger V\sp\dagger\left(
\mathcal{Q}^{-1} \left[
\frac{\partial}{\partial x_i}+\mathcal{H}\frac{z}{r^2}\,
x_i+\mathcal{H}\frac{\imath}{r^2} \left(
\boldsymbol{x}\times\boldsymbol{\nabla}\right)_i \right]
\right)V\mu\Big|_{z=-1}\sp{z=1}
\label{gaugeA}
\end{equation}
under a gauge transformation necessitates that we have 
\begin{align*}
h\sp{-1} A_i h +
h\sp{-1} \partial_i h &=
h\sp{-1}\left(\int_{-1}^1\mathrm{d}z\, \boldsymbol{v}\sp\dagger
\frac{\partial}{\partial x_i} \boldsymbol{v}\right) h=
h\sp{-1} A_i h +h\sp{-1} \left( \mu\sp\dagger V\sp\dagger \mathcal{Q}^{-1}V\mu\right)\Big|_{z=-1}\sp{z=1} 
\partial_i h \\
& \qquad +
h\sp{-1} \left( \mu\sp\dagger V\sp\dagger \mathcal{Q}^{-1}\mathcal{H}V\mu\right)\Big|_{z=-1}\sp{z=1} 
\frac{\imath}{r^2} \left(
\boldsymbol{x}\times\boldsymbol{\nabla}\right)_i h .
\end{align*}
This will follow again as a result of (\ref{gaugeconsistency}) and the requirement of the ADHMN theorem
\ref{ADHMN} that $V\mu$ is normalised by 
\begin{equation*}
1_2=\mu\sp{\dagger}\left(\int_{-1}^1dz\,V\sp{\dagger}V\right)\mu=
\mu\sp{\dagger}(\boldsymbol{x})
\left(V\sp{\dagger}(\boldsymbol{x},z)\mathcal{Q}\sp{-1}(\boldsymbol{x},z)V(\boldsymbol{x},z)\right)\big\vert_{z=-1}\sp{z=1}
\,\mu(\boldsymbol{x}).
\end{equation*}
The $z$-independence of the projectors $\mu$ together with the fact that the poles of $V$ lie only at the 
end points means that determining the projector is purely algebraic  and reduces to the  cancellation of poles.

Consideration of (\ref{gaugeconsistency}) leads to a new result
which is the analogue of Hitchin's hermitian form introduced in his description of monopoles \cite{hitchin_83}.
\begin{theorem}\label{constancythm} With the notation above,
\begin{equation}
\left(V\sp\dagger \mathcal{Q}^{-1}\mathcal{H} V\right)(z)=\textrm{constant},
\qquad
\left(W\sp\dagger \mathcal{Q}\mathcal{H} W\right)(z)=\textrm{constant},
\label{constancy}
\end{equation}
and consequently (\ref{gaugeconsistency}) holds true.
\end{theorem}

\begin{proof} As these are essentially the inverses of each other we shall only prove the first. Using
the notation and results of Appendix A we find
\begin{align*}
\frac{d}{dz}\left(V\sp\dagger \mathcal{Q}^{-1}\mathcal{H} V\right)&=
V\sp\dagger\left( \frac{d}{dz}\left[   \mathcal{Q}^{-1}\mathcal{H}\right] -(\mathcal{H}+\mathcal{F}) \mathcal{Q}^{-1}\mathcal{H}
-\mathcal{Q}^{-1}\mathcal{H}(\mathcal{H}+\mathcal{F})\right)V\\
&=
V\sp\dagger\left( \mathcal{D}( \mathcal{Q}^{-1})\mathcal{H}+\mathcal{Q}^{-1}(\mathcal{H}+\mathcal{F})\mathcal{H}
-\mathcal{Q}^{-1}\mathcal{H}(\mathcal{H}+\mathcal{F}) \right)V\\
&=
V\sp\dagger\left(\mathcal{H}+Q^{-1}[\mathcal{F},\mathcal{H}] \right)V\\
&=0
\end{align*}
where we have made explicit use of (\ref{panrel}) and (\ref{qprops}).
\end{proof}

Let us summarise what we have thus far. Given
\begin{itemize}
\item Nahm Data (needed to construct $\Delta\sp\dagger$ and $\mathcal{Q}$),
\item the fundamental matrices $V$ or $W$ to $\Delta\sp\dagger \boldsymbol{v}=0$
or $\Delta \boldsymbol{w}=0$, where $V=\left(W\sp\dagger\right)\sp{-1}$,
\end{itemize}
we can algebraically solve for the projector $\mu$ and consequently obtain the Higgs (\ref{higgs})
and gauge fields (\ref{gaugeA}). Although this is the strategy of our solution we shall ultimately show that we
only require this data up to a gauge transformation and that these are determined by the curve.

\section{Integrability and a lesser known Ansatz of Nahm}
In this section we shall express both the Nahm data and fundamental matrices $V$ and $W$ in terms of
the curve, clearly identifying what may be done explicitly and what is implicit. The key ingredient is
integrability which has not been visible so far. We first make some general remarks, then turn to
the construction of the Nahm data and fundamental matrices. 

Upon setting (with ${T_i}\sp\dagger=-T_i$, $T_4\sp\dagger =-T_4$)
\begin{equation*}
\alpha=T_4+\im T_3,\quad \beta = T_1+iT_2,
\quad
L=L(\zeta):=\beta -(\alpha+\alpha\sp\dagger)\zeta-\beta\sp\dagger \zeta^2, \quad
M=M(\zeta):=-\alpha-\beta\sp\dagger \zeta,
\end{equation*}
one finds
\begin{equation}
\begin{split}\label{integrability}
\dot{{T}_i} =[T_4,T_i]+\frac12\sum_{j,k=1}^3\epsilon_{ijk}[T_j(z),T_k(z)]
&\Longleftrightarrow
\dot L=[L,M]\\
&\Longleftrightarrow\quad
\left\{\begin{aligned}
\left[\dfrac{d }{dz}-\alpha,\beta\right]&=0,\\
\dfrac{d (\alpha+\alpha\sp\dagger)}{dz}&=[\alpha,\alpha\sp\dagger]+[\beta,\beta\sp\dagger].
\end{aligned}\right.
\end{split}
\end{equation}
The first equivalance here is that of a Lax pair, suggesting an underlying integrable system, while
the second equivalence expresses
Nahm's equations in the form of a complex and a real equation (respectively) \cite{donaldson84}. The complex Nahm
equation is readily solved,
\begin{equation}
\beta g = g \nu, \quad
\left( \frac{d}{dz}-\alpha\right)g=0 
\Longleftrightarrow
\beta = g\nu g\sp{-1} ,\quad  \alpha=\dot g g\sp{-1},
\end{equation}
where $\nu$ is constant and generically diagonal, $\nu=\diag(\nu_1,\ldots,\nu_n)$;  by 
conjugating\footnote{ $\tilde\beta =g(0)\sp{-1}\beta g(0)$,  $\tilde g(z) =g(0)\sp{-1} g(z)$,
$\tilde\alpha =g(0)\sp{-1}\alpha g(0)$.}  by
the constant matrix $g(0)$ we may assume $\beta(0)=\nu$ and $g(0)=1_n$.
It is the real equation that is more difficult. Define
\begin{equation}\label{defh}
h=g\sp\dagger g
\end{equation}
 then
\begin{equation}\label{defdhh}
\dot h h\sp{-1} =g\sp{\dagger }(\alpha+\alpha\sp\dagger)g\sp{\dagger \,-1}, \qquad h(0)=1_n,
\end{equation}
and the real equation yields the (possibly) nonabelian Toda equation
\begin{equation}\label{nonabtoda}
\frac{d}{dz} \left( \dot h h\sp{-1} \right)= \left[ h\nu h\sp{-1}, \nu\sp\dagger\right].
\end{equation}
In the monopole context Donaldson \cite{donaldson84} proved the existence of a solution for the real equation with
Nahm data. 
In terms of the Lax pair we are wishing to solve
\begin{align}
(L-\eta)U&=0,
\nonumber \\
\left[  \frac{d}{dz}+M\right]U&=0.\label{mscatter}
\end{align}
The characteristic equation $P(\eta,\zeta):=\det(\eta-L(\zeta))=0$ defines our spectral curve $\mathcal{C}$
which takes the form
\begin{equation}
P(\zeta,\eta):=\eta^n+a_1(\zeta)\eta^{n-1}+\ldots+a_n(\zeta)=0, \quad
\mathrm{deg}\, a_k(\zeta) \leq 2k\label{curve1}.
\end{equation}
For large $\zeta$ we see that 
$\det(\eta/\zeta^2-L/\zeta^2)\sim \prod_{i=1}\sp{n}(\eta/\zeta^2+\nu_i\sp\dagger)$ and so 
$\eta/\zeta\sim -\nu_i\sp\dagger\zeta$.
The curve $\mathcal{C}$ is an $n$-sheeted cover of $\mathbb{P}\sp1$
of  genus $g_\mathcal{C} =(n-1)^2$; we shall denote by $\{\infty_i\}_{i=1}\sp{n}$ the preimages of $\zeta=\infty$.
Setting $U=g\sp{\dagger\,-1}\Phi$ 
we  use the complex equation to transform (\ref{mscatter}) into a standard scattering equation
for $\Phi$,
\begin{equation}\label{standardscattering}
 \left[  \frac{d}{dz}-g\sp{\dagger }(\alpha+\alpha\sp\dagger)g\sp{\dagger \,-1}\right]\Phi=
\zeta \nu\sp\dagger\Phi.
\end{equation}
\lq\lq Standard\rq\rq here simply means that the matrix  $\zeta \nu\sp\dagger$
on the right-hand side is $z$-independent.
In terms of $h$ we have (\ref{defdhh}) and
\begin{align}
g\sp{\dagger}Lg\sp{\dagger\,-1}=h\nu h\sp{-1}-\dot h h\sp{-1}\zeta-\nu\sp\dagger\zeta^2.
\label{conjL}
\end{align}
The point to note is that we can solve the standard scattering equation
(\ref{standardscattering}) explicitly in terms of the function theory of $\mathcal{C}$ by what is known
as a Baker-Akhiezer function  \cite{krichever977a}, and so too $\dot h h\sp{-1}$ and  (the gauge transform)
$g\sp{\dagger}Lg\sp{\dagger\,-1}=\widehat\Phi\diag(\eta_1\ldots,\eta_n)\widehat\Phi\sp{-1}$, where
$\widehat\Phi:=\widehat\Phi(z,P)$ ($P\in \mathcal{C}$) is the fundamental matrix of solutions to (\ref{standardscattering}). Asymptotically the
$i$-th column of $\widehat\Phi$ behaves as $\exp(z \zeta \nu_i\sp\dagger)$ and the Baker-Akhiezer function
is defined by 
\begin{equation}\label{BAnorm}
\lim_{ P=P(\zeta,\eta)\rightarrow \infty_i  }\widehat\Phi(z, P)\exp(-z \zeta \nu_j\sp\dagger)=\diag(\delta_{ij}).
\end{equation}

In the monopole context Ercolani and Sinha  were the first to construct the Baker-Akhiezer function and
use this to study Nahm data \cite{ercolani_sinha_89}; this theory was (corrected and) extended in 
\cite{Braden2010d} with a needed generalization of the Abel-Jacobi map proved in \cite{Braden2008}. In the notation above these works explicitly construct the 
Baker-Akhiezer solutions to (\ref{standardscattering}) as well\footnote{In the
notation of Ercolani and Sinha we have $g\sp{\dagger\,-1 }=C$ and $\dot h h\sp{-1}=- Q_0(z)$.}
 as $\dot h h\sp{-1}$. Thus the Nahm data is only determined up to the gauge transformation
 $g\sp{\dagger\,-1}$ which is to satisfy the differential equation  (\ref{defdhh})  for $g$, where
 the left-hand side is to be viewed as specified. In the gauge where $T_4=0$, equivalently
$\alpha=\alpha\sp\dagger$, this takes the form
$$2 \dot g\sp\dagger  g\sp{\dagger\, -1}= \dot h h\sp{-1}$$
which we cannot solve in general. (In the Hitchin system context this is the analogue of
having solved the complex moment map equation ${\bar\partial}_A\phi=0$ and then being faced with
solving the real moment map equation $\mu(A)=F_A+[\phi,\phi\sp\dagger]=0$.)
In what follows we will not need the specific construction of these functions and will simply denote the solution to (\ref{standardscattering})  by $\Phi_{BA}$.

Although we can only solve for Nahm data up to a gauge transformation, what can we say about the solutions to either $\Delta\sp\dagger \boldsymbol{v}=0$ or $\Delta \boldsymbol{w}=0$? Nahm again
made the seminal ansatz; with some small changes better suited for the connections to integrability
we may encode this in terms of the following theorem (proven in Appendix B).
\begin{theorem}[Nahm \cite{nahm82c}]
Let $|s>$ be an arbitrarily normalized spinor not in $\ker
(1_2+\boldsymbol{\widehat{ u}}(\boldsymbol{x})\cdot\boldsymbol{\sigma})$,
with $\boldsymbol{\widehat{ u}}(\boldsymbol{x})$  a unit vector independent of $z$ and
$\boldsymbol{\sigma}=(\sigma_1,\sigma_2,\sigma_3)$. Then
\begin{equation}\label{Nansatz}
\boldsymbol{w}:=\boldsymbol{w}(\zeta)
=(1_2+\boldsymbol{\widehat{ u}}(\boldsymbol{x})\cdot\boldsymbol{\sigma})\,
e\sp{-\im z\left[(x_1-\im x_2)\zeta-\im x_3 -x_4\right]}|s>\otimes\,
U(z) 
\end{equation}
satisfies $\Delta \boldsymbol{w}=0$ if and only if
\begin{align}
0&=\left(L(\zeta)-\eta\right) U(z) ,\label{eqlax1}\\
0&=\left( \dfrac{\mathrm{d}}{\mathrm{d}z}+M(\zeta)
\right) U(z),\label{eqlax2} 
\end{align}
where
\begin{equation}
\eta=(x_2-\im x_1)-2 x_3\zeta-(x_2+\im x_1)\zeta^2,
\label{etadef}
\end{equation}
and $L(\zeta)$ and $M(\zeta)$, as above, satisfy the Lax equation $\dot L=[L,M]$.
\end{theorem}
The appendix also shows that the curve has a real structure, that is, it is invariant under the
anti-holomorphic involution:
\begin{equation}\label{invol}
\mathfrak{J}:\, (\zeta,\eta)\rightarrow
(-\frac1{\overline\zeta},-\frac{\overline\eta}{{\overline\zeta}\sp2}).
\end{equation}

Given our previous discussion our approach is natural: we may solve 
\begin{equation}\label{UBA}
U(z) =g\sp{\dagger\,-1}\,\Phi_{BA}
\end{equation}
in terms of the earlier (and unknown) gauge transformation $g\sp{\dagger\,-1}$ and the Baker-Akhiezer function. Prior to this
other workers had sought to explicitly perform the integrations. (This approach is reviewed in Appendix C.) 

It remains to describe the fundamental matrix $W$. 
Given a spectral curve 
and a position $\boldsymbol{x}$ we substitute the corresponding value of $\eta$ (given by
(\ref{etadef})) into (\ref{curve1}). This is an equation of degree $2n$ in $\zeta$ which we shall refer to as the \emph{Atiyah-Ward constraint}, this equation having appeared in their work. The $2n$ solutions $\zeta_j$
 give us $2n$ points on the curve $P_j:=(\zeta_j,\eta_j)$ ($j=1,\ldots,n$) where 
 $\eta_j$ again follows from $\boldsymbol{x}$  and (\ref{etadef}). These $2n$ points come in
 $n$ pairs of points related by the antiholomorphic involution $\mathfrak{J}$. To each point we have the associated values $\boldsymbol{\widehat{u}}(\zeta_j)$ and for each of these we solve for  $U(z)$ yielding a $2n\times 1$ matrix $\boldsymbol{w}(P_j)$. Taking each of the $2n$ solutions we obtain a $2n\times 2n$ matrix of solutions $W$. There may be non-generic points for which $\zeta_i=\zeta_j$ at which we
modify this discussion by taking a derivative $\boldsymbol{w}'(P_j)$;  these non-generic points correspond to points of bitangency of the spectral curve and appear in Hurtubise's study of the asymptotic behavious of 
the Higgs fileld \cite{hurtubise85a}.

Summarising the insights of this section: the Nahm data, the fundamental matrices $V$ or $W$ to $\Delta\sp\dagger \boldsymbol{v}=0$
or $\Delta \boldsymbol{w}=0$, where $V=\left(W\sp\dagger\right)\sp{-1}$,  and the Lax matrices
$L$, $M$ can be solved for in terms of the Baker-Akhiezer function \emph{up to
the  unknown gauge transformation $g\sp{\dagger\,-1}$}. In terms of this gauge transformation $h=g\sp{\dagger}g$
satisfies the nonabelian Toda equation (\ref{nonabtoda}).

\section{Constructing the gauge and Higgs Fields}

Having reviewed the general formalism and established the first new result Theorem \ref{constancythm}
we shall now extend this. Our aim is to establish
\begin{theorem}\label{reconstructionthm}
 Given a spectral curve $\mathcal{C}$ one may construct the gauge and Higgs fields
directly.
\end{theorem}

First, the form of (\ref{UBA}) and (\ref{Nansatz}) tells us that we may write
\[
W=1_2\otimes g\sp{\dagger\, -1} \,\widehat W,
\qquad 
V=1_2\otimes g \,\widehat V,
\qquad
\widehat V= \widehat W\sp{\dagger\, -1},
\]
where
\[
\boldsymbol{\widehat w}=(1_2+\boldsymbol{\widehat{ u}}(\boldsymbol{x})\cdot\boldsymbol{\sigma})\,
e\sp{-\im z\left[(x_1-\im x_2)\zeta-\im x_3 -x_4\right]}|s>\otimes\,
\Phi_{BA }.
\]
Here $\widehat W=({\boldsymbol{\widehat w}}_1,\ldots,\boldsymbol{\widehat w}_{2n})$ and $\widehat V$ are explicitly expressible in terms of the spectral curve $\mathcal{C}$.
Next observe that the operators appearing on the right-hand side of Proposition \ref{panagopolous}
take the form $\mathcal{Q}\sp{-1}\mathcal{O}$ and, because $g=g(z)$ is a function of $z$ only, these
indefinite integrals may be written as
\[
\mu\sp\dagger V\sp\dagger \mathcal{Q}^{-1} \mathcal{O}V\mu
=
\mu\sp\dagger  {\widehat V}\sp\dagger\left[ \left(1_2\otimes g\sp\dagger\right) \mathcal{Q}^{-1} \left(1_2\otimes g
\right) \right] \mathcal{O}{\widehat V}\mu
=
\mu\sp\dagger  {\widehat V}\sp\dagger \mathcal{Q}'\sp{-1}
\left(1_2\otimes h\right) \mathcal{O}{\widehat V}\mu
,
\]
where we recall that $h=g\sp\dagger g$. Here, using the definitions (\ref{pandefs}),
\[
\mathcal{Q}'=
 \left(1_2\otimes g\sp{\dagger}\right) \mathcal{Q} \left(1_2\otimes g\sp{\dagger\,-1}\right) 
 :=
\frac{1}{r^2} \mathcal{H}\mathcal{F}'\mathcal{H}-\mathcal{F}',
\qquad
 \mathcal{F}'=\imath \sum_{j=1}^3 \sigma_j\otimes  g\sp{\dagger} T_j g\sp{\dagger\,-1}
.
\]
Now using the definition of  $\alpha$, $\beta$ we have\footnote{$\sum_j\sigma_j\otimes T_j=
\frac12 (\sigma_1 +\imath \sigma_2)\otimes (T_1 -\imath T_2)+
\frac12 (\sigma_1 -\imath \sigma_2)\otimes (T_1 +\imath T_2)+\sigma_3\otimes T_3
= \sigma_-\otimes \beta+\sigma_3\otimes T_3-\sigma_+\otimes \beta\sp\dagger
$.
}
 from (\ref{conjL})  that
\[
 \mathcal{F}'=\begin{pmatrix} \frac12 \dot h h\sp{-1} & -\imath \nu\sp\dagger\\
 \imath h\nu h\sp{-1} &-\frac12 \dot h h\sp{-1}
 \end{pmatrix},
\]
and from our earlier remarks this may be reconstructed from the curve $\mathcal{C}$. 
Further $\widehat W$ and $\widehat V$ are determined by the Baker-Akhiezer function for the full range of $z\in[-1,1]$, so allowing their expansions at the end-points and the evaluation of the integrals of
Proposition \ref{panagopolous}.
At this stage we
see that the various impediments to reconstructing the gauge and Higgs  fields, notably the need for the Nahm
data to determine $\mathcal{Q}$ and only being able to determine $V$ up to a gauge transformation,
have combined into the one unkown (matrix) $h(z)$. Certainly the Baker-Akiezer function gives us
$\dot h h\sp{-1}$
from which one can in principle solve the ODE for $h(z)$, but our aim is to reconstruct the gauge and Higgs fields without having to do this integration. We shall now show how we may reconstruct
$h(z)$ and so prove the theorem.

\section{Determining $h(z)$}

To proceed we must be a little more precise about our spectral curve, the
nature of the coordinates $(\zeta,\eta)$ and the geometry of the Lax pair. 
In particular we need to specify the sections we are using in our construction.
Hitchin's construction \cite{hitchin_83} shows that the spectral curve
naturally lies in mini-twistor space: the spectral curve is an algebraic curve $\mathcal{C} \subset T\PP\sp1$.
If $\zeta$ is the inhomogeneous coordinate on the Riemann sphere then $(\zeta,\eta)$ are the standard local coordinates on $T\PP\sp1$ defined by $(\zeta,\eta)\rightarrow\eta\frac{d}{d\zeta}$. The anti-holomorphic involution (\ref{invol}) endows $T\PP\sp1$ with its standard real structure.

\subsection{Bundle Structure and Sections}
We wish to extend the definition of our Lax pair to yield a bundle on $\mathbb{P}\sp1$. Let $N$ and $S$
denote the standard patches on $\mathbb{P}\sp1$ (here $0\in N=\{\zeta\,|\, \zeta\ne\infty\} $) and denote
by the same expressions the corresponding open sets for $T\PP\sp1$. 
Take $L=L\sp{N}$, $M=M\sp{N}$, and we choose ($\zeta\sp{S}=1/\zeta$, $\eta\sp{S}=\eta/\zeta^2$)
\begin{equation}\label{bundleLax}
L\sp{S}(\zeta\sp{S})=\frac{L\sp{N}(\zeta)}{\zeta^2},\qquad 
M\sp{S}(\zeta\sp{S})= \alpha\sp\dagger-\frac{\beta}{\zeta}=M\sp{N}(\zeta)-\frac{L\sp{N}(\zeta)}{\zeta}.
\end{equation}
We have the reality conditions
\begin{equation}\label{bundlereality}
L\sp{S\, \dagger}(\zeta\sp{S})=-L\sp{N}(-1/\bar{\zeta}),
\qquad
M\sp{S\, \dagger}(\zeta\sp{S})=-M\sp{N}(-1/\bar{\zeta}).
\end{equation}
Away from the branch points of the spectral curve the eigenvectors\footnote{We will use
capital Roman letters $\{I,J,\ldots\}$ for rows and lower case Roman letters $\{i,j,\ldots\}$ for columns.}  $U^N_j$ are linearly independent and can be organized into a maximal rank matrix $U^N=(U^N_1, U^N_2, \ldots, U^N_n)$ so that
$L^N=U^N D^N\left(U^N\right)^{-1}$, with diagonal $D^N={\diag}(\eta^N_1,\eta^N_2,\ldots,\eta^N_n)$.
Viewing our curve $\mathcal{C}$ as an $n$-fold branched cover of $\mathbb{P}\sp{1}$ the columns $U^N_j$ correspond to the various sheets of the cover.
Similarly, $L^S=U^S D^S\left(U^S\right)^{-1}$. From (\ref{bundleLax}) we see
that the corresponding eigenvectors are proportional: 
$U^S=U^N F$ with diagonal $F={\rm diag}(F_1,F_2,\ldots,F_n)$. 
It also follows from (\ref{bundleLax}) and $({d}/{dz}+M\sp{N,S})U\sp{N,S}=0$ that
\begin{align}
\frac{d}{dz} F=\frac{D^N}{\zeta}F.
\end{align}
Thus $F$ is diagonal with its diagonal elements of the form $F_i(z,\zeta,\eta)=f_i(\zeta)e^{z\eta_i/\zeta}$.  For the case we consider the spectral curve has monodromy permuting all of its sheets (and thus all of its eigenvectors $U^{N,S}_j$); thus all functions the  $f_i(\zeta)$ must be equal and so
$F_i(z,\zeta,\eta)=f(\zeta)e^{z\eta_i/\zeta}$. Let us now focus on the $I$-th row of $U^N$: this defines a function $\Psi^N_I$ on $\mathcal{C}$ everywhere outside $\zeta=\infty$.  Similarly, the $I$-th row of $U^S$ defines a function $\Psi^S_I$ on $\mathcal{C}$ everywhere besides $\zeta=0.$ These two functions are related by $\Psi^S_I=\Psi^N_I\,f(\zeta)e^{z\eta/\zeta}$; in other words the pair of $I$-th rows of $U^N$ and $U^S$ define a section of the bundle over $\mathcal{C}$ with transition function $f(\zeta)\exp(z\eta/\zeta)$. Now by the Birkhoff-Grothendieck theorem a line bundle over $\mathbb{P}^1$ 
must be of the form $\mathcal{O}(r)$ for some integer $r$, and so by a change of trivialization we can take $f(\zeta)=1/\zeta^r$.

At this point we come to different choices in specifying the sections $U$. Let ${L}^z$ be the
line bundle with transition function $e^{z\eta/\zeta}$. For Hitchin $z\in[0,2]$ with $z=1$ the point where
${L}^1\otimes\pi\sp*\mathcal{O}(n-1)$ has $n$ sections with real structure. Hitchin's sections  \cite{hitchin_83} are defined just in terms of the two patches $N,S$ and $r=n-1$ in the above. 
The Baker-Akhiezer construction of \cite{ercolani_sinha_89,Braden2010d} describes the flow in terms of $s \in [-1,1]$ ($z=s+1$) and a 
line bundle $\mathcal{L}_\delta$ corresponding to a  nonspecial divisor 
$\delta=\sum_{i=1}\sp{g_{\mathcal{C}}+n-1}\delta_i$ 
of degree $g_{\mathcal{C}}+n-1=\deg\pi\sp*\mathcal{O}(n-1)$. 
The line bundles $\mathcal{L}_\delta$ and ${L}^1\otimes\pi\sp*\mathcal{O}(n-1)$ are linearly equivalent.
On $\mathcal{C}\setminus \{\infty_i\}_{i=1}\sp{n}$ the Baker-Akhiezer function is meromorphic with poles
in $\delta$; the transition functions for this line bundle are around each of the $\delta_i$. Thus in the $NS$
transition function we have $r=0$. In what follows we set 
\begin{equation}
\label{defF}
F=\diag\left(e\sp{z\eta_i/\zeta} /\zeta^l \right)
\end{equation}
where $l=n-1$ when describing Hitchin's choice of section and $l=0$ when describing the Baker-Akhiezer sections.

\subsection{Identifying $h(z)$}
The bundle structure just described enables us to identify $h(z)$. First observe that if
$\left[ d/dz +M\right]U=0$ then $\left[ d/dz -M\sp\dagger \right]U\sp{\dagger\, -1}=0$. Using this together
with (\ref{bundlereality}) shows 
$\left[ d/dz +M\sp{S}(\zeta\sp{S})\right]U\sp{N}(z,-1/\bar{\zeta})\sp{\dagger\, -1}=0$ and consequently
that 
\begin{equation}\label{defUS}
 \mathcal{D}\sp{S}(\zeta\sp{S}):=
U\sp{N}(z,-1/\bar{\zeta})\sp{\dagger} U\sp{S}(z,\zeta\sp{S})
\end{equation}
is a $z$-independent matrix. Now 
\begin{align*}
 \mathcal{D}\sp{S}(\zeta\sp{S}) D\sp{S}(\zeta\sp{S})
 &=U\sp{N}(z,-1/\bar{\zeta})\sp{\dagger} L\sp{S}(\zeta\sp{S}) U\sp{S}(z,\zeta\sp{S})
 =-\left(  L\sp{N}(-1/\bar{\zeta}) U\sp{N}(z,-1/\bar{\zeta}) \right)\sp\dagger U\sp{S}(z,\zeta\sp{S})\\
 &=- D\sp{N}(-1/\bar{\zeta})\sp\dagger  \mathcal{D}\sp{S}(\zeta\sp{S})
 = D\sp{S}(\zeta\sp{S})  \mathcal{D}\sp{S}(\zeta\sp{S}).
\end{align*}
Using the fact that the diagonal matrices $D\sp{N,\, S}$ have distinct entries 
for generic $(\zeta,\eta)$  we deduce that $\mathcal{D}\sp{S}$ is diagonal.  We similarly define the 
constant diagonal matrix
\begin{equation}\label{N}
\mathcal{D}\sp{N}(\zeta)=(-1)\sp{l} \,U\sp{S}(z,-1/{{\bar\zeta}\sp{S}})\sp\dagger U\sp{N}(z,\zeta).
\end{equation}
Then from $U\sp{S}=U\sp{N}F$ on an overlap we have
$$\mathcal{D}\sp{S}(\zeta\sp{S})
= \left( U\sp{S}(z,-\frac1{{\bar\zeta}\sp{S}}) F(z,-\frac1{{\bar\zeta}})\sp{-1} \right)\sp\dagger 
U\sp{N}(s,\zeta) F(z,\zeta)
=  (-1)\sp{l} F(z,-\frac1{\bar\zeta})\sp{\dagger\,-1} \mathcal{D}\sp{N} (\zeta)F(z,\zeta).
$$
Upon using that all of the matrices here are diagonal and that under the antiholomorphic involution
(\ref{invol}) 
$\eta/\zeta\rightarrow \bar\eta /{\bar\zeta}$ we have
$$\mathcal{D}\sp{S}= (-1)\sp{l}  f(\zeta)/f( -1/\bar\zeta)\sp\dagger  
\mathcal{D}\sp{N}= \frac{1}{\zeta\sp{2l}} \mathcal{D}\sp{N}.
$$

In the Hitchin setting $\{ \mathcal{D}\sp{N},\mathcal{D}\sp{S}\}$ yield a section of $\pi\sp\ast\mathcal{O}(2n-2)$;
such a section takes the form $c_0 \eta^{n-1}+c_1(\zeta)\eta^{n-2}+\ldots+c_{n-1}(\zeta)$
where $c_{l}(\zeta)$ is of degree $2l$ in $\zeta$. In our setting we get a regular function constant in $z$.

\subsection{Expressing $h(z)$ in terms of the Baker-Akhizer function}
Now  with $U=U\sp{N}=g(z)\sp{\dagger\, -1}\widehat\Phi (z,\zeta)$, where $\widehat\Phi (z,\zeta)$ is the $n\times n$ matrix
whose columns\footnote{
$
\widehat{\Phi}(z,\zeta):= ( \boldsymbol{\Phi}_1(z,P_1),\ldots,  
\boldsymbol{\Phi}_n(z,P_n)     )$, 
where $P_i=(\zeta,\eta_i)$. 
} are the Baker-Akhiezer functions for the preimages of $\zeta$, we have
\begin{align*}
\mathcal{D}\sp{N}(\zeta)=\left[ U\sp{N}(z,-\frac1{\bar\zeta}) \diag(e\sp{ z {\bar\eta_i}/{\bar\zeta}})\right]\sp{\dagger}g\sp{\dagger\, -1}\widehat\Phi (z,\zeta)
= \diag(e\sp{ z {\eta_i}/{\zeta}})\widehat \Phi(z, -\frac1{\bar\zeta})\sp{\dagger}h(z)\sp{-1}\widehat\Phi (z,\zeta).
\end{align*}
Now the left-handside is $z$-independent, so we may evaluate this at $z=0$ using $h(0)=1_n$ to give
$$
\mathcal{D}\sp{N}(\zeta)
=\widehat \Phi(0, -\frac1{\bar\zeta})\sp{\dagger}\widehat\Phi (0,\zeta).
$$
Employing this we obtain 
\begin{equation}
\label{hbaz}
h(z)=\widehat\Phi (z,\zeta)\diag(e\sp{ z {\eta_i}/{\zeta}}) \widehat\Phi (0,\zeta)\sp{-1}
\widehat \Phi(0, -\frac1{\bar\zeta})\sp{\dagger\, -1}\widehat \Phi(z, -\frac1{\bar\zeta})\sp{\dagger}.
\end{equation}
Again using the diagonality of $\mathcal{D}\sp{N}$ this expression shows that $h(z)\sp\dagger$ is obtained by the interchange of $\zeta\rightarrow -\frac1{\bar\zeta}$; but $h(z)$ is independent of $\zeta$ and so
we see that $h(z)$ is hermitian, as is required.
Finally, using the independence of $h(z)$ on the value of $\zeta$ being used,
 upon taking $\zeta$ to infinity and using (\ref{BAnorm}) and the hermiticity of $h$ we may rewrite this to give
\begin{theorem} \label{htheorem}
Let $\widehat \Phi(z, 0)$ be the Baker-Akhiezer function with $(I,i)$-entry ${\widehat\Phi}_I(z, 0_i)$, where
$0_i=\mathcal{J}(\infty_i)$, then
\begin{equation}
\label{hfinal}
h(z)=
\widehat \Phi(z, 0)\widehat \Phi(0, 0)\sp{-1}.
\end{equation}
\end{theorem}

We remark that the inverse of $\widehat \Phi(z, \zeta)$ may be constructed with the dual Baker-Akhiezer function; further, we see that this expression is independent of the ordering of sheets that is implicit in the
Baker-Akhiezer function.

At this stage we have established Theorem \ref{reconstructionthm} circumventing
the actual construction of Nahm data. To determine the Nahm data one
must solve for $h(z)=\widehat \Phi(z, 0)\widehat \Phi(0, 0)\sp{-1}=g\sp\dagger g$ where $h$ and $g$ must
satisfy (\ref{defdhh}). This constraint makes finding $g$ more difficult than simply a matrix factorization
problem; although Cholesky factorization enables one to factorize $h$ such a factorization is only
defined up to a unitary transformation which is determined by  (\ref{defdhh}). An interesting question is
whether the function theory of $\mathcal{C}$ might help in this determination.

\section{Example}
We shall give an example of Theorem \ref{htheorem} for the case $n=2$. The spectral curve for $n=2$
was constructed by Hurtubise \cite{hurtubise_83} and we shall employ the Ercolani-Sinha \cite{ercolani_sinha_89} form
\begin{equation}
0=\eta^2+\frac{K^2}4\left( \zeta^4+2(k^2-k'^2)\zeta^2+1\right),   \label{curve}
\end{equation}
where ${K}={K}(k)$ is a complete elliptic integral, and $\eta$ is related to the
spatial coordinates by{
\begin{equation}   
\eta=(x_2-\imath x_1)-2\zeta x_3-(x_2+\imath x_1)\zeta^2.
\label{AtiyahWard} 
\end{equation}
It was known that the solutions
\begin{equation}
\begin{split}\label{nahmsolution}
f_1(z)&={K}\,\frac{\mathrm{dn}\,{K}z
}{\mathrm{cn}\,{K}z}=\frac{\pi\theta_2 \theta_3 }{2}
\,\frac{\theta_3(z/2)}{\theta_2(z/2)},\quad f_2(z)={K}
k'\,\frac{\mathrm{sn}\,{K}z
}{\mathrm{cn}\,{K}z}=\frac{\pi\theta_3 \theta_4 }{2}
\,\frac{\theta_1(z/2)}{\theta_2(z/2)},\\
 f_3(z)&={K}k'\,\frac{1}{\mathrm{cn}\,{K}z}=\frac{\pi\theta_2 \theta_4 }{2}
\,\frac{\theta_4(z/2)}{\theta_2(z/2)}
\end{split}
\end{equation}
 to the spinning top equations 
$ \dot f_1=f_2\,f_3$ (and cyclic) gave solutions to the Nahm equations via $T_j(z)=\frac{\sigma_j}{2\imath}\, f_j(z)$ and the work of \cite{ercolani_sinha_89} and (with corrections in) \cite{Braden2010d} derived these from first principles. As we have noted the construction proceeds via the Baker-Akhiezer function but
this only yields the Nahm data up to a gauge transformation $g\sp{\dagger -1}$. In the works cited the differential equation for this gauge transformation was solved explicitly and we will compare the expression for $h(z)$ given by Theorem \ref{htheorem} with that obtained from the solution of the differential equation.
For higher charges the associated differential equation has not been solved.

The Baker-Akhiezer function for the problem at hand takes the form
\begin{equation}\label{bafull}
\Phi(z,P)=\chi(P)\begin{pmatrix}-\theta_3(\alpha(P))\theta_2(\alpha(P)-z/2) \\
\theta_1(\alpha(P))\theta_4(\alpha(P)-z/2)
\end{pmatrix}\,
\frac{e\sp{\beta_1(P) z}}{\theta_2(z/2)}
\end{equation}
where
\begin{equation}\label{bachi}
\chi(P)=
\frac{\theta_2(1/4)\theta_3(1/4)}{\theta_3(0)
\theta_1(\alpha(P)-1/4)\theta_4(\alpha(P)+1/4)},
\qquad 
\beta_1(P)=
\frac14\left\{ \frac{\theta_1'(\alpha(P))}{\theta_1(\alpha(P))}+
\frac{\theta_3'(\alpha(P))}{\theta_3(\alpha(P))}\right\} ,
\end{equation}
and we have the Abel map $\alpha(P)=\int_{\infty_1}\sp{P}d\zeta/4\eta$ for $P=(\zeta,\eta)$.
Then with $P_{1,2}$ corresponding to  $(\zeta,\pm\eta)$ we have
\begin{equation}
\begin{split}
\widehat \Phi(z,\zeta)=&
\begin{pmatrix}-\theta_3(\alpha(P_1))\theta_2(\alpha(P_1)-z/2) &
-\theta_3(\alpha(P_2))\theta_2(\alpha(P_2)-z/2)\\
\theta_1(\alpha(P_1))\theta_4(\alpha(P_1)-z/2) &
\theta_1(\alpha(P_2))\theta_4(\alpha(P_2)-z/2)
\end{pmatrix}\\
&\qquad \times
\begin{pmatrix}\frac {\chi(P_1) e\sp{\beta_1(P_1) z}} {\theta_2(z/2)}&0\\
0& \frac{\chi(P_2)  e\sp{\beta_1(P_2) z} }{\theta_2(z/2)}
\end{pmatrix}.
\end{split}
\end{equation}
To evaluate this for $0_{1,2}$ we note (from \cite{Braden2010d}) that  $\alpha(0_1)=-\tau/2$,
$\alpha(0_2)=-1/2$.
Using $\theta_4(\tau/2)=0=\theta_2(1/2)$ we may simplify (\ref{hfinal}) to
\begin{align*}
h(z)&= 
\begin{pmatrix}
\frac{ \theta_2(\alpha(0_1)-z/2) }{\theta_2[0_1]} &
-\frac{\theta_3[0_2] \theta_2(\alpha(0_2)-z/2)}{ \theta_1[0_2]\theta_4[0_2]  }\\
-\frac{\theta_1[0_1]\theta_4(\alpha(0_1)-z/2)}{\theta_3[0_1] \theta_2[0_1] } &
\frac{\theta_4(\alpha(0_2)-z/2)}{\theta_4[0_2]}
\end{pmatrix}
\begin{pmatrix}\frac {\theta_2 e\sp{\beta_1(0_1) z}} {\theta_2(z/2)}&0\\
0& \frac{\theta_2  e\sp{\beta_1(0_2) z} }{\theta_2(z/2)}
\end{pmatrix}
\\
&=
\begin{pmatrix}
\frac{ \theta_2(\tau/2+z/2) }{\theta_2(\tau/2)} &
\frac{\theta_3(1/2) \theta_2(1/2+z/2)}{ \theta_1(1/2)\theta_4(1/2)  }\\
\frac{\theta_1(\tau/2)\theta_4(\tau/2+z/2)}{\theta_3(\tau/2)\theta_2(\tau/2)} &
\frac{\theta_4(1/2+z/2)}{\theta_4(1/2)}
\end{pmatrix}
\begin{pmatrix}\frac {\theta_2 e\sp{\beta_1(0_1) z}} {\theta_2(z/2)}&0\\
0& \frac{\theta_2  e\sp{\beta_1(0_2) z} }{\theta_2(z/2)}
\end{pmatrix}
\\
&=
\begin{pmatrix}
\frac{ B(z/2)\theta_3(z/2)  }{ B(0)\theta_3 } &
-\frac{\theta_4 \theta_1(z/2)}{ \theta_2\theta_3 }\\
-\frac{B(z/2)\theta_4 \theta_1(z/2)}{ B(0)\theta_2\theta_3 }&
\frac{\theta_3(z/2)}{\theta_3}
\end{pmatrix}
\begin{pmatrix}\frac {\theta_2 e\sp{\beta_1(0_1) z}} {\theta_2(z/2)}&0\\
0& \frac{\theta_2  e\sp{\beta_1(0_2) z} }{\theta_2(z/2)}
\end{pmatrix}\\
\intertext{where $B(v)=\exp(-i\pi[v+\tau/4])$}
&=
\frac1{\theta_3 \theta_2(z/2)}
\begin{pmatrix} 
\theta_2 \theta_3(z/2) & -\theta_4 \theta_1(z/2) \\
-\theta_4 \theta_1(z/2) & \theta_2 \theta_3(z/2)
\end{pmatrix}
\begin{pmatrix}e\sp{\beta_1(0_1) z -i\pi z/2 }&0\\
0& e\sp{\beta_1(0_2) z} 
\end{pmatrix}.
\end{align*}
Now upon using the quasi-periodicity of the theta functions,
$$\frac{\theta_1'(-\tau/2)}{\theta_1(-\tau/2)}=i\pi -\frac{\theta_4'}{\theta_4},
\qquad \frac{\theta_3'(-\tau/2)}{\theta_3(-\tau/2)}=i\pi -\frac{\theta_2'}{\theta_2},$$
we find
$$\beta(0_1)=\frac{i\pi}2-\frac14\left\{ \frac{\theta_2'}{\theta_2}+ \frac{\theta_4'}{\theta_4} \right\}
=\frac{i\pi}2,
\qquad
\beta(0_2)=-\frac14\left\{ \frac{\theta_2'}{\theta_2}+ \frac{\theta_4'}{\theta_4} \right\}=0.
$$
Here we have used that $\theta_{2,4}$ are even and so $\theta_{2,4}'=0$. Thus
$$h(z)=\frac1{\theta_3 \theta_2(z/2)}
\begin{pmatrix} 
\theta_2 \theta_3(z/2) & -\theta_4 \theta_1(z/2) \\
-\theta_4 \theta_1(z/2) & \theta_2 \theta_3(z/2)
\end{pmatrix}.
$$
In terms of (\ref{nahmsolution})
and using $K=\pi \theta_3^2/2$ we arrive at
\begin{equation}\label{charge2hfinal}
h(z)=\frac1{K}\begin{pmatrix}f_1&-f_2\\-f_2&f_1 \end{pmatrix},
\qquad
h\sp{-1}(z)=\frac1{K}\begin{pmatrix}f_1&f_2\\f_2&f_1 \end{pmatrix}.
\end{equation}

We wish to compare (\ref{charge2hfinal}) with
$h\sp{-1}=C\sp\dagger C$
where $C$ is the solution to the differential equation
$$C\sp{-1}\dot C= \dot C C\sp{-1}=\frac{f_3}2\begin{pmatrix}0&1\\1&0\end{pmatrix}.$$
Then
$$C=C\sp\dagger=C\sp{T}=\left(\begin{array}{cc}F(z)&G(z)\\G(z)&F(z)
\end{array} \right),$$
with solution
\[ F=\cosh\left(\int_0\sp{z} f_3(s)ds/2 \right)=\left[ p(z)+1/p(z)\right]/2,\quad
G=\sinh\left(\int_0\sp{z} f_3(s)ds/2 \right)=\left[ p(z)-1/p(z)\right]/2,\]
where\footnote{
Here we have made use of
\begin{equation*}
\int \frac{du}{\mathrm{cn}\,u } =\frac{1}{k'}\,\mathrm{ln}\frac{
\mathrm{dn}u + k' \mathrm{sn} u }{\mathrm{cn} u}.\label{primitive}
\end{equation*}
}
\begin{align*}
p(z)&=\exp\left(\int_0\sp{z} f_3(s)ds/2 \right)=\exp\left(k' K\int_0\sp{z} \frac{ds}{\mathrm{cn}\,{K}z} \right)
= \left[ \frac{ \mathrm{dn}\,{K}z + k'\mathrm{sn}\,{K}z}{ \mathrm{cn}\,{K}z}
\right]\sp{1/2}.
\end{align*}
{Now
\begin{align} \label{FGrels}
G^2(z)&=\frac12\left( \frac{\mathrm{dn} (Kz;k)}{\mathrm{cn}
(Kz;k)}-1\right)=\frac12\left( \frac{f_1}{K}-1\right),&\quad 2F(z)G(z)&=k'\frac{\mathrm{sn}(Kz;k)}{\mathrm{cn}(Kz;k)}=\frac{f_2}{K}, \\
F^2-G^2&=1,& F^2+G^2&=\frac{f_1}{K},
 \nonumber
\end{align}
and consequently
$C^2(z)=\dfrac{1}{K}\begin{pmatrix}f_1&f_2\\   f_2&f_1
\end{pmatrix},
$
so verifying (\ref{charge2hfinal}).}

\section*{Acknowledgements}
HWB thanks Sergey Cherkis, Derek Harland and Paul Sutcliffe for helpful conversations.

\appendix
\section{The Panagopoulos Formulae}\label{section:panagopoulos}

Panagopolous \cite{panagopo83} introduced three formulae to evaluate integrals appearing in the ADHMN construction giving a proof for one of these; his other formulae were checked against a calculation on
one axis. We shall prove and slightly modify here the Panagopolous formulae, extending his method to the case $x_4,T_4$ possibly nonzero and correcting an ordinary differential operator by a partial differential
operator (which reduces to the former on an axis).

Panagopolous's approach to evaluate integrals 
$$\int dz\, \boldsymbol{{v}}_a\sp\dagger
\mathcal{A}\boldsymbol{{v}}_b$$
for any
given operator $\mathcal{A}$ and any two solutions
$\boldsymbol{{v}}_{a,b}$ of $\Delta^{\dagger}\boldsymbol{v}=0$ 
is to seek  an
operator $\mathcal{B}$ such that
\begin{equation}
\boldsymbol{{v}}_a\sp\dagger
\mathcal{A}\boldsymbol{{v}}_b=\frac{\mathrm{d}}{\mathrm{d}z} \left(\boldsymbol{{v}}_a\sp\dagger
\mathcal{B}\boldsymbol{{v}}_b\right),
\end{equation}
so giving a primitive. Using (\ref{deltadagger}, \ref{pandefs}) then if $\boldsymbol{v}$ is any solution of
$\Delta^{\dagger}\boldsymbol{v}=0$ we have that
\begin{align*}
1_{2n}\frac{\mathrm{d}}{\mathrm{d}z}\boldsymbol{v}=\left[\im x_4+T_4-(\mathcal{H}+\mathcal{F})\right]\boldsymbol{v},
\end{align*}
whence
\begin{align*}
\boldsymbol{{v}}_a\sp\dagger\mathcal{A}\boldsymbol{{v}}_b
&=\frac{\mathrm{d} \boldsymbol{{v}}_a\sp\dagger}{\mathrm{d}z}
\mathcal{B}\boldsymbol{{v}}_b
+\boldsymbol{{v}}_a\sp\dagger
\frac{\mathrm{d}\mathcal{B}}{\mathrm{d}z}\boldsymbol{{v}}_b
+\boldsymbol{{v}}_a\sp\dagger\mathcal{B}
\frac{\mathrm{d}\boldsymbol{{v}}_b}{\mathrm{d}z}
=\boldsymbol{{v}}_a\sp\dagger\left(
 \frac{\mathrm{d}\mathcal{B}}{\mathrm{d}z} -[T_4,\mathcal{B}]-(\mathcal{H}+\mathcal{F})\mathcal{B}-\mathcal{B}(\mathcal{H}+\mathcal{F})
  \right)\boldsymbol{{v}}_b.
\end{align*}
The diagonal term with $x_4$ cancels here. Thus we seek to relate the operators $\mathcal{A}$ and
$\mathcal{B}$ by
\[\mathcal{A}=
\mathcal{D} (\mathcal{B})
:=
\frac{\mathrm{d}\mathcal{B}}{\mathrm{d}z}
 -[T_4,\mathcal{B}]
-(\mathcal{H}+\mathcal{F})\mathcal{B}-\mathcal{B}(\mathcal{H}+\mathcal{F}).
\]
where we have defined the operator $\mathcal{D}$.

We will use the following relations (recall $T_4$ is shorthand for $1_2\otimes T_4$),
\begin{align}\label{qprops}
 \mathcal{Q}\mathcal{H}= \mathcal{H}\mathcal{F}- \mathcal{F}\mathcal{H}
& =\left[\mathcal{H},\mathcal{F}\right]=-\mathcal{H} \mathcal{Q}, 
 \qquad
\mathcal{H}=  \mathcal{Q}\sp{-1}\left[\mathcal{H},\mathcal{F}\right],
\qquad
[\mathcal{H}, T_4]=0,
\\
x_{i}\left({\partial_i \mathcal{H}}\right)&=\mathcal{H},
\quad\qquad
\mathcal{H}^2=r^2 1_{2n}, \nonumber
\\
\mathcal{F}^2&=-1_2\otimes
\sum_{i=1}^3T_iT_i-\im\sum_{i,j,k=1}^3\epsilon_{ijk}\,\sigma_k\otimes
T_iT_j ,\nonumber
\\
\frac{d\mathcal{F}}{\mathrm{d}z} -[T_4, \mathcal{F}]&=\im\sum_{i,j,k=1}^3\epsilon_{ijk}\, \sigma_k\otimes
T_iT_j  .  \nonumber
\intertext{Therefore}
0&= \left[\mathcal{F}^2+\frac{\mathrm{d}\mathcal{F}}{\mathrm{d}z} -[T_4, \mathcal{F}],\mathcal{H}\right].
\label{Fcomm}
\end{align}
We now establish the integrals in (\ref{panagopolous}).
\begin{proposition} \label{panagopo1} With 
$
\mathcal{Q}=\frac{1}{r^2} \mathcal{H}\mathcal{F}\mathcal{H}-\mathcal{F} 
$
then 
\begin{equation}
\mathcal{D}(\mathcal{Q}^{-1})=\frac{\mathrm{d}\mathcal{Q}^{-1}}{\mathrm{d}z}
-[T_4, \mathcal{Q}^{-1}]
-(\mathcal{H}+\mathcal{F})\mathcal{Q}^{-1}
-\mathcal{Q}^{-1}(\mathcal{H}+\mathcal{F}) =1_{2n} \label{panrel}
\end{equation}
and we have the antiderivative (\ref{pannorm})
\begin{equation*}
\int \mathrm{d}z\, \boldsymbol{{v}}_a\sp\dagger \boldsymbol{{v}}_b
= \boldsymbol{{v}}_a\sp\dagger \mathcal{Q}^{-1}
\boldsymbol{{v}}_b.
\end{equation*}
\end{proposition}
\begin{proof}
In this case $\mathcal{A}=1_{2n}$ and the left-hand side  of (\ref{panrel}) may be rewritten as follows
\begin{align*}
\mathcal{Q}^{-1} &\left[ -\frac{\mathrm{d}}{\mathrm{d}z}\left(
\frac{1}{r^2}\mathcal{H}\mathcal{F}\mathcal{H}-\mathcal{F}
\right)
+[T_4, \mathcal{Q}]
-(\mathcal{H}+\mathcal{F})\left(
\frac{1}{r^2}\mathcal{H}\mathcal{F}\mathcal{H}-\mathcal{F} \right)
-\left(
\frac{1}{r^2}\mathcal{H}\mathcal{F}\mathcal{H}-\mathcal{F}
\right)(\mathcal{H}+\mathcal{F}) \right]\mathcal{Q}^{-1}\\
=&\mathcal{Q}^{-1}
\left[-\frac{1}{r^2}\mathcal{H}
\left(\frac{\mathrm{d}\mathcal{F}}{\mathrm{d}z} -[T_4, \mathcal{F}]\right)
\mathcal{H}+\frac{\mathrm{d}\mathcal{F}}{\mathrm{d}z}
-[T_4, \mathcal{F}]
-\frac{1}{r^2}(\mathcal{H}+\mathcal{F})\mathcal{H}\mathcal{F}\mathcal{H}
-\frac{1}{r^2}\mathcal{H}\mathcal{F}\mathcal{H}(\mathcal{H}+\mathcal{F})\right.\\
&\qquad \qquad 
+(\mathcal{H}+\mathcal{F})\mathcal{F}+\mathcal{F}(\mathcal{H}+\mathcal{F})\bigg]\mathcal{Q}^{-1}\\
=&\mathcal{Q}^{-1}
\left[\frac{1}{r^2}\mathcal{H}\left(\mathcal{F}^2+1_2\otimes\sum_{i=1}^3
T_iT_i\right)\mathcal{H}-\left(\mathcal{F}^2+1_2\otimes\sum_{i=1}^3T_iT_i\right)\right.\\
&\left.\qquad\ -\frac{1}{r^2}\mathcal{H}^2\mathcal{F}\mathcal{H}-\frac{1}{r^2}
\mathcal{F}\mathcal{H}\mathcal{F}\mathcal{H}+\mathcal{H}\mathcal{F}+\mathcal{F}^2
-\frac{1}{r^2}\mathcal{H}\mathcal{F}\mathcal{H}^2-\frac{1}{r^2}
\mathcal{H}\mathcal{F}\mathcal{H}\mathcal{F}+\mathcal{F}\mathcal{H}+\mathcal{F}^2
\right]\mathcal{Q}^{-1}
\end{align*}
Upon noting that
\begin{align*}
\mathcal{Q}^2&=\left(\frac{1}{r^2}\mathcal{H}\mathcal{F}\mathcal{H}-\mathcal{F}\right)^2
=\frac{1}{r^4}\mathcal{H}\mathcal{F}\mathcal{H}^2\mathcal{F}\mathcal{H}+\mathcal{F}^2
-\frac{1}{r^2}\mathcal{H}\mathcal{F}\mathcal{H}\mathcal{F}
-\frac{1}{r^2}\mathcal{F}\mathcal{H}\mathcal{F}\mathcal{H}\\
&=\frac{1}{r^2}\mathcal{H}\mathcal{F}^2\mathcal{H}
-\frac{1}{r^2}\mathcal{H}\mathcal{F}\mathcal{H}\mathcal{F}
-\frac{1}{r^2}\mathcal{F}\mathcal{H}\mathcal{F}\mathcal{H}+\mathcal{F}^2
\end{align*}
and performing the appropriate cancellations we obtain the necessary
result.

\end{proof}

\begin{proposition}\label{panagopo2} With $\mathcal{Q}$ as previously given we have the antiderivative
\begin{equation}
\int \mathrm{d}z\, z \boldsymbol{{v}}_a\sp\dagger \boldsymbol{{v}}_b
= \boldsymbol{{v}}_a\sp\dagger \mathcal{S}
\boldsymbol{{v}}_b,
\qquad{\textrm{where}}\qquad
 \mathcal{S}=\mathcal{Q}^{-1} \left( z+\mathcal{H}\,\frac{x_{i}}{r^2}\frac{\partial}{\partial x_{i}}   \right)
.\label{parel2}
\end{equation}
\end{proposition}
\begin{proof}
Set $ \mathcal{S}_1=\mathcal{Q}^{-1} z$ and 
$\mathcal{S}_2=\mathcal{Q}^{-1} \mathcal{H}\frac{x_{i}}{r^2}\frac{\partial}{\partial x_{i}}$.
Then
\begin{align*}
\mathcal{D}(\mathcal{S}_1)&=\frac{\mathrm{d}}{\mathrm{d}z}\left(z\mathcal{Q}^{-1}\right)
 -z[T_4, \mathcal{Q}\sp{-1}]\
-z
(\mathcal{H}+\mathcal{F})\mathcal{Q}^{-1}-z\mathcal{Q}^{-1}(\mathcal{H}
+\mathcal{F})
=z\mathcal{D}(\mathcal{Q}^{-1})+\mathcal{Q}^{-1}
\\
&=z1_{2n}+\mathcal{Q}^{-1}.
\end{align*}
Further, again using (\ref{panrel}),
\begin{align*}
r^2\,\mathcal{D}(\mathcal{S}_2)
&=\frac{\mathrm{d}}{\mathrm{d}z}\left(\mathcal{Q}^{-1} \mathcal{H}x_{i}\frac{\partial}{\partial x_{i}}\right)
-[T_4, \mathcal{Q}^{-1} \mathcal{H}x_{i}\frac{\partial}{\partial x_{i}}]
-(\mathcal{H}+\mathcal{F})\mathcal{Q}^{-1} \mathcal{H}x_{i}\frac{\partial}{\partial x_{i}}
-\mathcal{Q}^{-1} \mathcal{H}x_{i}\frac{\partial}{\partial x_{i}}(\mathcal{H}+\mathcal{F})
\\
&=\mathcal{D}(\mathcal{Q}^{-1})\mathcal{H}x_{i}\frac{\partial}{\partial x_{i}}
\;\;+\mathcal{Q}^{-1}(\mathcal{H}+\mathcal{F}) \mathcal{H}x_{i}\frac{\partial}{\partial x_{i}}-\mathcal{Q}^{-1} \mathcal{H}x_{i}\frac{\partial}{\partial x_{i}}(\mathcal{H}+\mathcal{F})\\
&= \mathcal{H}x_{i}\frac{\partial}{\partial x_{i}}+
\mathcal{Q}^{-1}(\mathcal{H}+\mathcal{F}) \mathcal{H}x_{i}\frac{\partial}{\partial x_{i}}
-\mathcal{Q}^{-1} \mathcal{H(\mathcal{H}+\mathcal{F})}x_{i}\frac{\partial}{\partial x_{i}}
-\mathcal{Q}^{-1}\mathcal{H}x_{i}\left({\partial_i \mathcal{H}} \right)
\\
&=\mathcal{Q}^{-1}\left[  
\mathcal{Q}\mathcal{H}  +(\mathcal{H}+\mathcal{F}) \mathcal{H}
- \mathcal{H(\mathcal{H}+\mathcal{F})}
 \right] 
x_{i}\frac{\partial}{\partial x_{i}}
-\mathcal{Q}^{-1}\mathcal{H}x_{i}\left({\partial_i \mathcal{H}} \right)\\
&=-\mathcal{Q}^{-1}r^2
\end{align*}
Upon combining these,
$\mathcal{D}(\mathcal{S})=
\mathcal{D}(\mathcal{S}_1+\mathcal{S}_2)=z1_{2n}
$
and the result follows.

\end{proof}
We note that we may write $\mathcal{S}_2)$ more symmetrically if needed:
\[ 
\mathcal{S}_2=\frac12\left[
\overleftarrow{ \frac{\partial}{\partial x_{i}}  }\, \frac{x_{i}}{r^2}
 \mathcal{H}\mathcal{Q}^{-1}+
\mathcal{Q}^{-1} \mathcal{H}\frac{x_{i}}{r^2}  \overrightarrow{\frac{\partial}{\partial x_{i}}}
\right].
\]

\noindent{\bf{Remark:}} In \cite{panagopo83} Panagopolous gave the ordinary differential operator
$ \mathcal{S}'=\mathcal{Q}^{-1} \left( z+2\mathcal{H}\,\frac{\mathrm{d}}{\mathrm{d} r^2}   \right)$ for
which
$\mathcal{D}(\mathcal{S}')=z1_{2n}+\mathcal{Q}^{-1}\left[ \frac{\mathrm{d}\mathcal{H}}{\mathrm{d}r^2}, \mathcal{H}   \right]$; as we are acting on functions and gauge transformations that are not just a function
of $r$ we have replaced this with our partial differential operator $\mathcal{S}$.


\begin{proposition}\label{panagopo3} With $\mathcal{Q}$ as previously given we have the antiderivative
\begin{align}
\int \boldsymbol{{v}}_a\sp\dagger\frac{\partial}{\partial
x_i}\boldsymbol{{v}}_b \mathrm{d}z
&=\boldsymbol{{v}}_a\sp\dagger\mathcal{Q}^{-1} \left[
\frac{\partial}{\partial x_i}+\mathcal{H}\frac{z}{r^2}\,
x_i+\mathcal{H}\frac{\imath}{r^2} \left(
\boldsymbol{x}\times\boldsymbol{\nabla}\right)_i \right]\boldsymbol{{v}}_b
.\label{parel3}
\end{align}
\end{proposition}

\begin{proof}
Let $L=L_1+L_2+L_3$ with
\[ L_1=\mathcal{Q}^{-1}\frac{\partial}{\partial x_i},\quad
L_2= \mathcal{Q}^{-1}\mathcal{H}\frac{z}{r^2}\,x_i,\quad
L_3=\mathcal{Q}^{-1}\mathcal{H}\frac{\imath}{r^2}\left( \boldsymbol{x}\times
\boldsymbol{\nabla} \right)_i .
\]
We compute $\mathcal{D}(L_i)$, $i=1,2,3$. First
\begin{align*}
\mathcal{D}(L_1)&=\frac{\mathrm{d}}{\mathrm{d}z}\left( \mathcal{Q}^{-1} \frac{\partial
}{\partial x_i} \right)
-[T_4, \mathcal{Q}^{-1} \frac{\partial
}{\partial x_i} ]
-(\mathcal{H}+\mathcal{F})\mathcal{Q}^{-1}\frac{\partial}{\partial
x_i}-\mathcal{Q}^{-1}\frac{\partial }{\partial
x_i}(\mathcal{H}+\mathcal{F})\\
&=\mathcal{D}(\mathcal{Q}^{-1})
\frac{\partial}{\partial x_i}-\mathcal{Q}^{-1}\frac{\partial
\mathcal{H}}{\partial x_i}
=1_{2n}\frac{\partial}{\partial x_i}-\mathcal{Q}^{-1}\frac{\partial
\mathcal{H}}{\partial x_i}
\end{align*}
where we use Proposition \ref{panagopo1}. Next
\begin{align*}
\mathcal{D}(L_2)&=\frac{\mathrm{d}}{\mathrm{d}z}\left( \mathcal{Q}^{-1}
\mathcal{H}\frac{z}{r^2}x_i \right)
-[T_4, \mathcal{Q}^{-1}\mathcal{H}\frac{z}{r^2}x_i]
-(\mathcal{H}+\mathcal{F})
\mathcal{Q}^{-1} \mathcal{H}\frac{z}{r^2}x_i-\mathcal{Q}^{-1}
\mathcal{H}\frac{z}{r^2}x_i(\mathcal{H}+\mathcal{F})\\
&=
\mathcal{D}({Q}^{-1})\mathcal{H}\frac{z}{r^2}x_i
+\mathcal{Q}^{-1}(\mathcal{H}+\mathcal{F})\mathcal{H}\frac{z}{r^2}x_i
-\mathcal{Q}^{-1}\mathcal{H}\frac{z}{r^2}x_i(\mathcal{H}+\mathcal{F})
+\mathcal{Q}^{-1}\mathcal{H}\frac{x_i}{r^2}\\
&=
\mathcal{H}\frac{z}{r^2}x_i
+\mathcal{Q}^{-1}\left[\mathcal{F}, \mathcal{H}\right]\frac{z}{r^2}x_i
+\mathcal{Q}^{-1}\mathcal{H}\frac{x_i}{r^2}\\
&=\mathcal{Q}^{-1}\mathcal{H}\frac{x_i}{r^2}
\end{align*}
upon using  (\ref{qprops}).
Next we calculate $\mathcal{D}(L_3)$ for $i=1$. We have
\begin{align*}
\mathcal{D}(L_3)&=\frac{\mathrm{d}}{\mathrm{d}z}\left[ \mathcal{Q}^{-1}
\mathcal{H}\frac{\imath}{r^2}\left( x_2\frac{\partial}{\partial
x_3}-x_3\frac{\partial}{\partial x_2} \right) \right]
-\left[T_4, \mathcal{Q}^{-1}
\mathcal{H}\frac{\imath}{r^2}\left( x_2\frac{\partial}{\partial
x_3}-x_3\frac{\partial}{\partial x_2} \right)\right]
\\
&\quad -(\mathcal{H}+\mathcal{F})\left[ \mathcal{Q}^{-1}
\mathcal{H}\frac{\imath}{r^2}\left( x_2\frac{\partial}{\partial
x_3}-x_3\frac{\partial}{\partial x_2} \right) \right]
-\left[ \mathcal{Q}^{-1} \mathcal{H}\frac{\imath}{r^2}\left(
x_2\frac{\partial}{\partial x_3}-x_3\frac{\partial}{\partial x_2}
\right) \right](\mathcal{H}+\mathcal{F})\\
&=\mathcal{D}({Q}^{-1})\mathcal{H}\frac{\imath}{r^2}\left(
x_2\frac{\partial}{\partial x_3}-x_3\frac{\partial}{\partial x_2}
\right)\\
&\quad +\mathcal{Q}^{-1}(\mathcal{H}+\mathcal{F})\mathcal{H}\frac{\imath}{r^2}\left(
x_2\frac{\partial}{\partial x_3}-x_3\frac{\partial}{\partial x_2}
\right)
-\mathcal{Q}^{-1}\mathcal{H}\frac{\imath}{r^2}\left(
x_2\frac{\partial}{\partial x_3}-x_3\frac{\partial}{\partial x_2}
\right)(\mathcal{H}+\mathcal{F})\\
&=\mathcal{H}\frac{\imath}{r^2}\left( x_2\frac{\partial}{\partial
x_3}-x_3\frac{\partial}{\partial x_2}
\right)+\mathcal{Q}^{-1}\left[\mathcal{F},\mathcal{H}\right]\frac{\imath}{r^2}\left(
x_2\frac{\partial}{\partial x_3}-x_3\frac{\partial}{\partial x_2}
\right)\\
&\qquad-\mathcal{Q}^{-1}\mathcal{H}\frac{\imath}{r^2}\left(
x_2\frac{\partial}{\partial x_3}-x_3\frac{\partial}{\partial x_2}
\right)\mathcal{H}\\
&=-\mathcal{Q}^{-1}\mathcal{H}\frac{\imath}{r^2}\left(
x_2\frac{\partial}{\partial x_3}-x_3\frac{\partial}{\partial x_2}
\right)\mathcal{H}.
\end{align*}
again using  (\ref{qprops}). Altogether we have
\begin{align*}
\mathcal{D}(L)=1_{2n}\frac{\partial}{\partial
x_1}+\mathcal{Q}^{-1}\left\{-\frac{\partial \mathcal{H}}{\partial
x_1}+\mathcal{H}\frac{x_1}{r^2}-\mathcal{H}\frac{\imath}{r^2} \left(
x_2\frac{\partial}{\partial x_3}-x_3\frac{\partial}{\partial x_2}
\right)\mathcal{H}\right\}
\end{align*}
Multiplying the expression in the parentheses by $-r^2$ gives
\begin{align*}
-\left\{\cdot\right\}r^2&=-\sigma_1\otimes 1_2(x_1^2+x_2^2+x_3^2)
+\sigma_1\otimes 1_2x_1^2+\sigma_2\otimes 1_2x_1
x_2+\sigma_3\otimes 1_2x_1 x_3\\
&\quad
-\imath(\sigma_1\otimes 1_2 x_1 +\sigma_2\otimes 1_2 x_2+\sigma_3\otimes 1_2 x_3  )
\times (x_2\sigma_3\otimes 1_2-x_3\sigma_2\otimes 1_2)
\end{align*}
which vanishes by standard relations, proving the result.

\end{proof}

\section{The Nahm Ansatz}
Introduce the shorthand
\begin{align}
R_j=T_j+\imath x_j 1_n,\quad j=1,2,3, \qquad 
\gamma =-\im \left[(x_1-\im x_2)\zeta-\im x_3 \right].
\end{align}
Upon substituting (\ref{Nansatz}) into $\Delta\boldsymbol{w}=0$ we find
\[
\begin{split}
0&=|s>\otimes\left(\im \dfrac{\mathrm{d}}{\mathrm{d}z}-\im T_4+
\im\gamma
+
\boldsymbol{\widehat{ u}}\cdot\boldsymbol{R}\right)\boldsymbol{\widehat{ w}(z)}
\\
&\qquad +
\sigma_k||s>\otimes\left(\im {\widehat{ u}}\sp{k}\left(\dfrac{\mathrm{d}}{\mathrm{d}z}-T_4+\gamma\right)
+R\sp{k}+
\im(\boldsymbol{R}\times\boldsymbol{\widehat{ u}})\sp{k}\right)\boldsymbol{\widehat{ w}}(z).
\end{split}
\]
and so we require
\begin{align}
0&=\left(\im \dfrac{\mathrm{d}}{\mathrm{d}z}-\im T_4+\im\gamma+
\boldsymbol{\widehat{ u}}\cdot\boldsymbol{R}\right)\boldsymbol{\widehat{ w}}(z),\label{nahm1}\\
0&=\mathcal{L}_k\boldsymbol{\widehat{ w}}(z):=\left(\im {\widehat{ u}}\sp{k}\left(\dfrac{\mathrm{d}}{\mathrm{d}z}-T_4 +\gamma\right)+R\sp{k}+
\im(\boldsymbol{R}\times\boldsymbol{\widehat{ u}})\sp{k}\right)\boldsymbol{\widehat{ w}}(z).\label{nahm2a}
\end{align}
The consistency of these equations imposes various constraints.
First consider
\begin{align*}
\left[\mathcal{L}_1,\mathcal{L}_2
\right]&=(\im {\widehat{ u}}\sp1+{\widehat u}\sp2{\widehat{ u}}\sp3)
\left(\dot{T_2}-[T_4,T_1]-[T_3,T_1]\right)-(\im {\widehat{ u}}\sp2-{\widehat u}\sp1{\widehat u}\sp3)
\left(\dot{T_1}-[T_4,T_2]-[T_2,T_3]\right)
\\
&\qquad-(1-({\widehat{ u}}\sp3)\sp2)\left(\dot{T_3}-[T_4,T_3]-[T_1,T_2]\right)+(1-\boldsymbol{\widehat{ u}}\cdot\boldsymbol{\widehat u})\,\left(\dot{T_3}-[T_4,T_3]\right).
\end{align*}
Thus provided $\boldsymbol{\widehat{ u}}(\boldsymbol{x})$ is a unit
vector and the $T_i$'s satisfy the Nahm equations we have
consistency of the equations $\mathcal{L}_k\boldsymbol{\widehat
w(z)}=0$.

At this stage we introduce a convenient parametrization
(reflected in Hitchin's minitwistor construction). Let
$\boldsymbol{ y}\in \mathbb{C}\sp3$ be a null vector. We may
consider $\boldsymbol{ y}\in \mathbb{P}\sp2$ and parameterize
$\boldsymbol{ y}$ as
\begin{equation}\label{defhy} \boldsymbol{
y}=\left(\frac{1+\zeta^2}{2\im},\frac{1-\zeta^2}2,-\zeta \right).
\end{equation}
Then
$$\boldsymbol{ y}\cdot \boldsymbol{\overline
y}=\frac{(1+|\zeta|\sp2)\sp2}2,\qquad \boldsymbol{
y}\cdot \boldsymbol{ y}=0.$$
The signs here have been chosen so that
$$L(\zeta):=2\im \boldsymbol{ y}\cdot\boldsymbol{T}=
(T_1+\im T_2)-2\im T_3\,\zeta+(T_1-\im T_2)\,\zeta\sp2.$$ In due
course we will see this to be our Lax matrix. Set
\begin{equation}\label{defhu}
\boldsymbol{\widehat{u}}=\boldsymbol{\widehat{u}}(\zeta):=\im\,\frac{\boldsymbol{ y}\times \boldsymbol{\overline
y}}{\boldsymbol{ y}\cdot \boldsymbol{\overline
y}}=\frac{1}{1+|\zeta|^2}\left(\im(\zeta-\overline{\zeta}),\;(\zeta+\overline{\zeta}),
1-|\zeta|^2\right).
\end{equation}
Then
$$\boldsymbol{\widehat{
u}}\times\boldsymbol{ y}=-\im\boldsymbol{ y},\qquad
\boldsymbol{\widehat{ u}}\times\boldsymbol{\overline
y}=\im\boldsymbol{\overline y}.$$ The three vectors
$\Re(\boldsymbol{ y}),\Im(\boldsymbol{ y})$ and $ \boldsymbol{\widehat{u}}$
form an orthogonal basis in $\mathbb{R}^3$ with
$|\boldsymbol{\widehat{ u}}|=1$, whence any $\boldsymbol{
v}\in\mathbb{R}^3$ may be written as
\begin{equation*}\label{basisuy}\boldsymbol{ v}=\boldsymbol{\widehat{
u}}\,(\boldsymbol{\widehat{ u}}\cdot \boldsymbol{v})+
\boldsymbol{\overline y}\left(\frac{\boldsymbol{ y}\cdot
\boldsymbol{v}}{ \boldsymbol{ y}\cdot\boldsymbol{\overline
y}}\right) + \boldsymbol{ y}\left(\frac{\boldsymbol{\overline
y}\cdot \boldsymbol{v}}{ \boldsymbol{ y}\cdot\boldsymbol{\overline
y}}\right).
\end{equation*}
In particular,
\begin{equation}\label{basisucy}\boldsymbol{ v}+\im
\boldsymbol{ v}\times\boldsymbol{\widehat{ u}}=
\boldsymbol{\widehat{u}}\,(\boldsymbol{\widehat{ u}}\cdot \boldsymbol{v})+2
\boldsymbol{\overline y}\left(\frac{\boldsymbol{ y}\cdot
\boldsymbol{v}}{ \boldsymbol{ y}\cdot\boldsymbol{\overline
y}}\right) .
\end{equation}
 We record that
\begin{align*}
\overline{\boldsymbol{ y}(\zeta)}&=-{\overline\zeta}\sp2\, \boldsymbol{ y}(-1/{\overline
\zeta}),\qquad\boldsymbol{\widehat{ u}}(-1/{\overline
\zeta})=-\boldsymbol{\widehat{ u}}(\zeta),\\
\boldsymbol{\widehat{ u}}&=(-\im\zeta\sp{-1},\zeta\sp{-1},-1)-\frac{2\boldsymbol{ y}}
{\zeta(1+|\zeta|\sp2)}
=(\im\zeta,\zeta,1)+\frac{2{\overline\zeta}\boldsymbol{ y}}
{1+|\zeta|\sp2},\\
\boldsymbol{\widehat{ u}}\cdot\boldsymbol{T}&=
-\im\left[(T_1+\im T_2)\zeta\sp{-1}-\im T_3\right]-
\frac{2\boldsymbol{ y}\cdot\boldsymbol{ T}}
{\zeta(1+|\zeta|\sp2)}=
\im\left[(T_1-\im T_2)\zeta-\im T_3\right]+
\frac{2{\overline\zeta}\boldsymbol{ y}\cdot\boldsymbol{ T}}
{1+|\zeta|\sp2}.
\end{align*}

Parameterizing $\boldsymbol{\widehat{ u}}$ as above and using
(\ref{basisucy}) we may write
\begin{align*}
\im \boldsymbol{\widehat{ u}}\left(\dfrac{\mathrm{d}}{\mathrm{d}z}-T_4 +\gamma \right)
+\boldsymbol{R}+
\im\boldsymbol{R}\times\boldsymbol{\widehat{ u}}&=\im \boldsymbol{\widehat{ u}}\left(\dfrac{\mathrm{d}}{\mathrm{d}z}-T_4  +\gamma \right)
+
\boldsymbol{\widehat{
u}}\,(\boldsymbol{\widehat{ u}}\cdot \boldsymbol{R})+2
\boldsymbol{\overline y}\left(\frac{\boldsymbol{ y}\cdot
\boldsymbol{R}}{ \boldsymbol{ y}\cdot\boldsymbol{\overline
y}}\right)\\& =\boldsymbol{\widehat{ u}}\left(\im  \left(\dfrac{\mathrm{d}}{\mathrm{d}z}-T_4
+\gamma \right)
+
\boldsymbol{\widehat{ u}}\cdot\boldsymbol{R}\right)+2
\boldsymbol{\overline y}\left(\frac{\boldsymbol{ y}\cdot
\boldsymbol{R}}{ \boldsymbol{ y}\cdot\boldsymbol{\overline
y}}\right)
\end{align*}
and as a consequence (\ref{nahm1}, \ref{nahm2a}) are equivalent to
\begin{align}
0&=\left(\im \dfrac{\mathrm{d}}{\mathrm{d}z}-\im T_4+ \im \gamma+
\boldsymbol{\widehat{ u}}\cdot\boldsymbol{R}\right)\boldsymbol{\widehat{ w}}(z),\label{nahm1b}\\
0&=\left(\boldsymbol{ y}\cdot
\boldsymbol{R}\right)\boldsymbol{\widehat{ w}}(z).\label{nahm2b}
\end{align}
The remaining consistency to be checked is then
\begin{align*}
\left[\im \dfrac{\mathrm{d}}{\mathrm{d}z}-\im T_4 +\im \gamma+
\boldsymbol{\widehat{ u}}\cdot\boldsymbol{R},\,
\boldsymbol{ y}\cdot
\boldsymbol{R}
\right]
&=
\im {\boldsymbol{y}}\cdot\left({\dot{\boldsymbol{T}}}- [T_4, \boldsymbol{T}]\right)
+
\left[\boldsymbol{\widehat{ u}}\cdot\boldsymbol{T},\,
{\boldsymbol{y}\cdot{\boldsymbol{T}}}
\right]
=0,
\end{align*}
which upon use of $\boldsymbol{\widehat{ u}}\times\boldsymbol{
y}=-\im\boldsymbol{ y}$ is equivalent to Nahm's equations.

Equally from
$$\boldsymbol{\widehat{ u}}\cdot\boldsymbol{R}=
-\im\left[(R_1+\im R_2)\zeta\sp{-1}-\im R_3\right]-
\frac{2\boldsymbol{ y}\cdot\boldsymbol{ R}}
{\zeta(1+|\zeta|\sp2)}=
\im\left[(R_1-\im R_2)\zeta-\im R_3\right]+
\frac{2{\overline\zeta}\boldsymbol{ y}\cdot\boldsymbol{ R}}
{1+|\zeta|\sp2}$$ we may write the equations as
\begin{align*}
0&=\left( \dfrac{\mathrm{d}}{\mathrm{d}z}-T_4 +\gamma
+\left[(R_1-\im R_2)\zeta-\im R_3\right]\right)\boldsymbol{\widehat{ w}}(z)
=\left( \dfrac{\mathrm{d}}{\mathrm{d}z}+M
\right)\boldsymbol{\widehat{ w}}(z),\\
0&=\left(\boldsymbol{ y}\cdot
\boldsymbol{R}\right)\boldsymbol{\widehat{ w}}(z),
\end{align*}
{where} \begin{equation} M(\zeta):=(T_1-\im T_2)\zeta-\im T_3 -T_4.
\end{equation}

The equations we have obtained are just the Lax equations
\begin{align}
0&={2\im}\left(\boldsymbol{ y}\cdot
\boldsymbol{R}\right)\boldsymbol{\widehat{ w}}(z)
=\left(L(\zeta)-\eta\right)\boldsymbol{\widehat{ w}}(z), \nonumber
\\
0&=\left( \dfrac{\mathrm{d}}{\mathrm{d}z}+M
\right)\boldsymbol{\widehat{ w}}(z),\nonumber\\
\intertext{for which $\dot L=[L,M]$ and where our construction defines $\eta$ to be}
\eta&={2\boldsymbol{y}\cdot{\boldsymbol{x}}}=(x_2-\im x_1)-2 x_3\zeta-(x_2+\im x_1)\zeta^2.
\nonumber 
\end{align}

From the first of these we see that
$$0=\det\left(L(\zeta)-\eta\right),$$
which gives the equation of the spectral curve $\mathcal{C}$. Upon
using $\overline{\boldsymbol{ y}(\zeta)}=-{\overline\zeta}\sp2\,
\boldsymbol{ y}(-1/{\overline \zeta})$ we see from
$$0=\det\left(L(\zeta)-\eta\right)\sp\dagger
=\det\left(L(\zeta)\sp\dagger-\overline{\eta}\right)=
\det\left({2\im\overline{\boldsymbol{ y}(\zeta)}\cdot
\boldsymbol{T}}-\overline{\eta}\right)=
\det\left({-2\im{\overline\zeta}\sp2{\boldsymbol{ y}(-1/{\overline
\zeta})}\cdot
\boldsymbol{T}}-\overline{\eta}\right)
$$
that the spectral curve is invariant under
$$\mathfrak{J}:\, (\zeta,\eta)\rightarrow
(-\frac1{\overline\zeta},-\frac{\overline\eta}{{\overline\zeta}\sp2}).$$
The spectral curve then has the form
\begin{equation}
P(\eta,\zeta):=\eta^n+a_1(\zeta)\eta^{n-1}+\ldots+a_n(\zeta)=0, \quad
\mathrm{deg}\, a_k(\zeta) \leq 2k,
\end{equation}
and the genus of $\mathcal{C}$ is $g=(n-1)^2$.

\section{Direct Integration of the Lax Equations}
When Nahm gave his ansatz he also suggested a direct integration of the resulting Lax
equations; Panagopoulos \cite[\S3] {panagopo83} tried implementing this for the charge $2$ case.
We will review here this approach and conclude by contrasting it with the proposed Baker-Akhiezer approach we have described.

Underlying the algebro-geometric description of Lax pairs and the spectral curve is that 
$\dim_{\mathbb{C}}\ker\left[L(\zeta)-\eta\right]=1$. Thus if $\boldsymbol{f}_i$ is an  eigenvector
of $L$ with associated eigenvalue $\eta_i$ then so too is $\boldsymbol{f}_i h_i(z)$. The idea is
that given a solution of (\ref{eqlax1}) to find an appropriate $h_i(z)$ such that (\ref{eqlax2}) is also
satisfied. To implement this we need to know both an eigenvector $\boldsymbol{f}_i $ and its derivative
to subsequently obtain a differential equation for $h_i(z)$. To determine these we make some simple
observations:
\begin{enumerate}
\item Suppose $L:=L(\zeta)$  satisfies the Lax equation $\dot L=[L,M]$. Then for any $(\zeta,\eta)$
$$\det(L(\zeta)-\eta)\, 1_n= (L(\zeta)-\eta) \Adj(L(\zeta)-\eta)$$
is constant, where $\Adj$ denotes the adjugate matrix. Differentiating this yields
\begin{equation*}
\left (L(\zeta)-\eta\right) 
\left(\frac{\mathrm{d}}{\mathrm{d}z} \adj (L(\zeta)-\eta)  -[
\adj (L(\zeta)-\eta) ,M]\right)=0
\end{equation*}
and using the generic
invertibility of $L(\zeta)-\eta$ we obtain
\begin{equation}\label{diffadj}
\frac{\mathrm{d}}{\mathrm{d}z} \adj (L(\zeta)-\eta)  =[
\adj (L(\zeta)-\eta) ,M].
\end{equation}

\item  For any constant vector $\boldsymbol{\nu}$ let $P=(\zeta,\eta_i)$ lie on the spectral curve. Then
\begin{equation}
\label{nevw}
\boldsymbol{f}_i=\adj(L-\eta_i)\boldsymbol{\nu}
\end{equation} 
is an eigenvector of $L(\zeta)$ with eigenvalue $\eta_i$ as
$$(L-\lambda_i )\boldsymbol{f}_i
=(L-\eta_i)\adj(L-\eta_i)\boldsymbol{\nu} =\det(L-\eta_i)\boldsymbol{\nu}=0.$$ 

\end{enumerate}
Using these observations we seek a function $h_i(z)$ such that 
$F=\adj(L-\eta_i)\boldsymbol{\nu}h_i(z)$ is solution of (\ref{eqlax2}).
Now
\begin{align*} \dot
F&=[\adj(L-\eta_i),M]\boldsymbol{\nu} h_i(z)+\adj(L-\eta_i)\boldsymbol{\nu} \dot
h_i(z)
\\&
=-MF+F h_i\sp{-1}\dot h_i+\adj(L-\eta_i)M\boldsymbol{\nu}
h_i(z)\end{align*} 
Taking the inner product with an arbitrary vector
$\boldsymbol{\mu}$ and requiring $F$ satisfy (\ref{eqlax2}) yields the differential equation
\begin{equation}\label{directintegration}
h_i\sp{-1}\frac{\mathrm{d}h_i}{\mathrm{d}z}=-
\frac{\boldsymbol{\mu}\sp{T}\adj(L-\eta_i)M\boldsymbol{\nu}}{\boldsymbol{\mu}\sp{T}\adj(L-\eta_i)\boldsymbol{\nu}}.
\end{equation}

We see then that to employ this approach we need both the Nahm data to be able to determine the right-hand side of (\ref{directintegration}) and to be able to integrate the subsequent rational expressions. For
$n=2$ this is possible. For higher charges we can only construct Nahm data up to the gauge transformation described in the main text. Let us imagine we make such a gauge transformation
and then use the approach of this appendix to solve  (\ref{eqlax1},\ref{eqlax2}). Now the
right-hand side of (\ref{directintegration}) is in principle known and an integration is required.
By way of contrast the Baker-Akhiezer approach allows the construction
of the gauge transform of both $L$ and the solution of (\ref{eqlax2}) without further integration.

\bibliographystyle{plain}

\def\cprime{$'$} \def\cdprime{$''$} \def\cprime{$'$}

\end{document}